%% file: main.tex
\documentclass{LMCS}

\def\dOi{11(2:16)2015}
\lmcsheading%
{\dOi}
{1--34}
{}
{}
{Nov.~\phantom08, 2014}
{Jun.~30, 2015}
{}

\ACMCCS{[{\bf Theory of computation}]: Logic---Logic and verification}


\usepackage{tikz}
\usepackage{hyperref}
\usepackage{amsmath}
\usepackage[subnum]{cases}
\usepackage{amssymb,stmaryrd}
\usepackage{listings}
\usepackage{verbatim}
\usepackage{color}
\usepackage{graphicx}
\usepackage{algorithm}
\usepackage{algpseudocode}
\usepackage{multirow}
\usepackage{pdfpages}
\usepackage{multibib}
\usepackage{wrapfig}
\usetikzlibrary{shapes}
\usetikzlibrary{positioning}
\usetikzlibrary{arrows}
\usetikzlibrary{calc}

\newcites{add}{Additional References}

\long\def\symbolfootnote[#1]#2{\begingroup%
\def\thefootnote{\fnsymbol{footnote}}\footnote[#1]{#2}\endgroup}

\input{macros-dave}

\title[Permissive Controller Synthesis for Probabilistic Systems]{Permissive Controller Synthesis \\ for Probabilistic Systems}

\author[K.~Dr\"ager]{Klaus~Dr\"ager\rsuper a}	
\address{{\lsuper a}EECS, Queen Mary, University of London, UK}	
\email{k.draeger@qmul.ac.uk}

\author[V.~Forejt]{Vojt\v{e}ch~Forejt\rsuper b}	
\address{{\lsuper{b,c,e}}Department of Computer Science, University of Oxford, UK}	
\email{\{vojtech.forejt,marta.kwiatkowska,mateusz.ujma\}@cs.ox.ac.uk}

\author[M.~Kwiatkowska]{Marta~Kwiatkowska\rsuper c}	
\address{\vspace{-18 pt}}

\author[D.~Parker]{David~Parker\rsuper d}	
\address{{\lsuper d}School of Computer Science, University of Birmingham, UK}	
\email{d.a.parker@cs.bham.ac.uk}

\author[M.~Ujma]{Mateusz Ujma\rsuper e}	
\address{\vspace{-18 pt}}

\thanks{The authors are in part supported by
ERC Advanced Grant VERIWARE
and EPSRC projects EP/K038575/1, 
EP/F001096/1 and 
EP/M023656/1. 
Vojt\v{e}ch Forejt is also affiliated with Faculty of Informatics, Masaryk University, Czech Republic.
}
\begin{document}

\begin{abstract}
We propose novel controller synthesis techniques for probabilistic systems
modelled using stochastic two-player games: one player acts as a controller,
the second represents its environment, and probability is used to
capture uncertainty arising due to, for example, unreliable sensors or faulty system components.
Our aim is to generate robust controllers that are
resilient to unexpected system changes at runtime,
and flexible enough to be adapted if additional constraints need to be imposed.
We develop a \emph{permissive} controller synthesis framework,
which generates \emph{multi-strategies} for the controller,
offering a choice of control actions to take at each time step.
We formalise the notion of permissivity using penalties,
which are incurred each time a possible control action is disallowed by a multi-strategy.
Permissive controller synthesis aims to generate a multi-strategy that minimises these penalties,
whilst guaranteeing the satisfaction of a specified system property.
We establish several key results about the optimality of multi-strategies
and the complexity of synthesising them.
Then, we develop methods to perform permissive controller synthesis using mixed integer linear programming
and illustrate their effectiveness on a selection of case studies.
\end{abstract}

\maketitle

\section{Introduction}

Probabilistic model checking is used to automatically verify systems with stochastic behaviour.
Systems are modelled as, for example, Markov chains, Markov decision processes,
or stochastic games, and analysed algorithmically to 
verify quantitative properties specified in temporal logic.
Applications include checking the safe operation of fault-prone systems
(``the brakes fail to deploy with probability at most $10^{-6}$'')
and establishing guarantees on the performance of,
for example, randomised communication protocols
(``the expected time to establish connectivity between two devices never exceeds 1.5 seconds'').

A closely related problem is that of \emph{controller synthesis}.
This entails constructing a model of some entity that can be controlled
(e.g., a robot, a vehicle or a machine) and its environment,
formally specifying the desired behaviour of the system,
and then generating, through an analysis of the model,
a controller that will guarantee the required behaviour.
In many applications of controller synthesis, a model of the system is inherently probabilistic.
For example, a robot's sensors and actuators may be unreliable,
resulting in uncertainty when detecting and responding to its current state;
or messages sent wirelessly to a vehicle may fail to be delivered with some probability.

In such cases, the same techniques that underly probabilistic model checking
can be used for controller synthesis.
For, example, we can model the system as a Markov decision process (MDP),
specify a property $\phi$ in a probabilistic temporal logic such as PCTL or LTL,
and then apply probabilistic model checking.
This yields an optimal \emph{strategy} (policy) for the MDP,
which instructs the controller as to which action should be taken
in each state of the model in order to
guarantee that $\phi$ will be satisfied.
This approach has been successfully applied in a variety of application domains,
to synthesise, for example:
control strategies for robots~\cite{LWAB10},
power management strategies for hardware~\cite{FKN+11},
and efficient PIN guessing attacks against hardware security modules~\cite{Ste06}.

Another important dimension of the controller synthesis problem
is the presence of uncontrollable or adversarial 
aspects of the environment.
We can take account of this by phrasing the system model as a \emph{game}
between two players, one representing the controller and the other the environment.
Examples of this approach include controller synthesis
for surveillance cameras~\cite{OTMW11}, autonomous vehicles~\cite{CKSW13}
or real-time systems~\cite{BCD+07}.
In our setting, we use (turn-based) stochastic two-player games,
which can be seen as a generalisation of MDPs where decisions are made by two distinct players.
Probabilistic model checking of such a game
yields a strategy for the controller player which guarantees
satisfaction of a property $\phi$, regardless of the actions of the environment player.

In this paper, we tackle the problem of synthesising \emph{robust} and \emph{flexible} controllers,
which are resilient to unexpected changes in the system at runtime.
For example, one or more of the actions that the controller can choose at runtime
might unexpectedly become unavailable, or additional constraints may be imposed
on the system that make some actions preferable to others.
One motivation for our work is its applicability to
model-driven runtime control of adaptive systems~\cite{CGKM12},
which uses probabilistic model checking in an online fashion to adapt or reconfigure
a system at runtime in order to guarantee the satisfaction of
certain formally specified performance or reliability requirements.


We develop novel, \emph{permissive} controller synthesis techniques
for systems modelled as stochastic two-player games.
Rather than generating \emph{strategies}, which specify 
a single action to take at each time-step,
we synthesise \emph{multi-strategies}, which specify 
multiple possible actions. As in classical controller synthesis,
generation of a multi-strategy is driven by a formally specified
quantitative property: we focus on probabilistic reachability
and expected total reward properties. 
The property must be guaranteed to hold, whichever of the
specified actions are taken and regardless of the behaviour of the environment.
Simultaneously, we 
aim to synthesise multi-strategies that are as \emph{permissive} as possible,
which we quantify by assigning \emph{penalties} to actions.
These are incurred when a multi-strategy disallows (does not make available) a given action.
Actions can be assigned different penalty values
to indicate the relative importance of allowing them.
Permissive controller synthesis amounts to finding a multi-strategy
whose total incurred penalty is minimal, or below some given threshold.


We formalise the permissive controller synthesis problem
and then establish several key theoretical results.
In particular, we show that
randomised multi-strategies are strictly more powerful than deterministic ones,
and we prove that the permissive controller synthesis problem is NP-hard for either class.
We also establish upper bounds, showing that the problem is in NP and PSPACE
for the deterministic and randomised cases, respectively.
%

Next, we propose practical methods for synthesising multi-strategies
using mixed integer linear programming (MILP) \cite{Schrijver}.
We give an exact encoding for deterministic multi-strategies
and an approximation scheme (with adaptable precision) for the randomised case.
For the latter, we prove several additional results
that allow us to reduce the search space of multi-strategies.
The MILP solution process works incrementally,
yielding increasingly permissive multi-strategies,
and can thus be terminated early if required.
This is well suited to scenarios where time is limited,
such as online analysis for runtime control, as discussed above,
or ``anytime verification''~\cite{Sha10}.
Finally, we implement our techniques and
evaluate their effectiveness on a range of case studies.

\vspace*{0.5em}
This paper is an extended version of~\cite{DFK+14},
containing complete proofs, optimisations for MILP encodings
and experiments comparing performance under two different MILP solvers.

\subsection{Related Work}

Permissive strategies in \emph{non-}stochastic games
were first studied in~\cite{BJW02} for parity objectives,
but permissivity was defined solely by comparing enabled actions.
Bouyer et al.~\cite{BDMR09} showed that optimally permissive memoryless strategies
exist for reachability objectives and expected penalties,
contrasting with our (stochastic) setting, where they may not.
The work in~\cite{BDMR09} also studies penalties given as mean-payoff and discounted reward functions,
and~\cite{BMOU11} extends the results to the setting of parity games.
None of~\cite{BJW02,BDMR09,BMOU11} consider stochastic games or even randomised strategies, 
and they provide purely theoretical results.
As in our work, Kumar and Garg~\cite{KG01} consider control
of stochastic systems by dynamically disabling events; however, rather
than stochastic games, their models are essentially Markov chains, which the
possibility of selectively disabling branches turns into MDPs.
\cite{Rad12} studies games where the aim of one opponent is to ensure properties of systems 
against an opponent who can modify the system on-the-fly by removing some transitions.
%

Finally, although tackling a rather different problem (counterexample generation),
\cite{WJABK12} is related in that it also uses MILP to solve probabilistic verification problems.


\section{Preliminaries}\label{sec:prelims}

We denote by $\Dist(X)$ the set of discrete probability distributions over a set $X$.
A \emph{Dirac} distribution is one that assigns probability 1 to some $s\in X$.
The \emph{support} of a distribution $d\in\Dist(X)$ 
is defined as $\support{d} \rmdef \{x\in X \mid d(x) > 0\}$.



\subsection{Stochastic Games}
In this paper, we use \emph{turn-based stochastic two-player games},
which we often refer to simply as \emph{stochastic games}. 
%
A stochastic game takes the form $\sgame=\gametuple$,
where $S\rmdef S_\Diamond\cup S_\Box$ is a finite set of states, each associated with
player $\Diamond$ or $\Box$, $\sinit\in S$ is an initial state, $\act$ is a finite set of actions and
$\delta : S {\times} \act \rightarrow \Dist(S)$ is a (partial) probabilistic transition function such that
the distributions assigned by $\delta$ only select elements of $S$ with rational probabilities.

An MDP is a stochastic game 
in which either $S_\Diamond$ or $S_\Box$ is empty.
Each state $s$ of a stochastic game $\sgame$ has a set 
of \emph{enabled} actions,
given by $\enab{s}\rmdef\{a\in\act \,|\, \delta(s,a)\mbox{ is defined}\}$.
The 
unique player $\eitherplayer\in\{\Diamond,\Box\}$ such that $s\in S_{\eitherplayer}$
picks an action $a\in\enab{s}$ to be taken in state $s$.
Then, the next state is determined randomly according
to the distribution $\delta(s,a)$, i.e., a transition to state $s'$ occurs with probability $\delta(s,a)(s')$.

A \emph{path} through $\sgame$ is a (finite or infinite) sequence $\omega=s_0 a_0 s_1 a_1 \dots$,
where $s_i\in S$, $a_i\in\enab{s_i}$ and $\delta(s_i,a_i)(s_{i+1})>0$ for all $i$.
We denote by $\ipaths_s$ the set of all infinite paths starting in state $s$.
For a player $\eitherplayer\in\{\Diamond,\Box\}$, we denote by
$\fpaths^\eitherplayer$ the set of all finite paths starting in any state and ending in a state from $S_\eitherplayer$.

A \emph{strategy}
$\sigma:\fpaths^\eitherplayer \rightarrow \dist(\act)$
for player $\eitherplayer$
of $\sgame$ is a resolution of the choices of actions in each state
from $S_\eitherplayer$ based on the execution so far,
such that only enabled actions in a state are chosen with non-zero probability.
In standard fashion \cite{KSK76}, a pair of strategies $\sigma$ and
$\pi$ for $\Diamond$ and $\Box$ induces, for any state $s$, a probability measure
$\Prb_{\sgame,s}^{\sigma,\pi}$ over $\ipaths_s$.
%
A strategy $\sigma$ is \emph{deterministic} if $\sigma(\omega)$ is a Dirac distribution
for all $\omega$, and \emph{randomised} if not.
In this work, we focus purely on \emph{memoryless} strategies,
where $\sigma(\omega)$ depends only on the last state of $\omega$,
in which case we define the strategy as a function 
$\sigma:S_\eitherplayer \rightarrow \dist(\act)$.
We write $\strats_\sgame^\eitherplayer$ for the set of all (memoryless)
player $\eitherplayer$ strategies in $\sgame$.



\subsection{Properties and Rewards}

In order to synthesise controllers, we need a formal description of their required properties.
In this paper, we use two common classes of properties:
\emph{probabilistic reachability} and \emph{expected total reward},
which we will express in an extended version of the temporal logic PCTL~\cite{HJ94}.

For probabilistic reachability, we write properties of the form $\phi = \calP_{\bowtie p}[\,\future g\,]$,
where $\bowtie\,\,\in\{\leq,\geq\}$, $p\in[0,1]$ and $g\subseteq S$ is a set of target states,
meaning that the probability of reaching a state in $g$ satisfies the bound $\bowtie p$.
Formally, for a specific pair of strategies
$\ctrlstrat\in\strats_\sgame^\Diamond, \envstrat\in\strats_\sgame^\Box$ for $\sgame$,
the probability of reaching $g$ under $\ctrlstrat$ and $\envstrat$ is
\[
\Prb_{\sgame,\sinit}^{\ctrlstrat,\envstrat}(\future g) \rmdef
\Prb_{\sgame,\sinit}^{\ctrlstrat,\envstrat}
(\{s_0 a_0 s_1 a_1 \dots\in\ipaths_\sinit\,|\, s_i\in g \mbox{ for some } i\}).
\]
We say that $\phi$ is satisfied under $\ctrlstrat$ and $\envstrat$,
denoted $\sgame,\ctrlstrat,\envstrat \models \phi$, if
$\Prb_{\sgame,\sinit}^{\ctrlstrat,\envstrat}(\future g) \bowtie p$.


We also augment stochastic games with 
\emph{reward structures}, which are functions of the form
$r:S\times\act\to\Qset_{\ge 0}$ mapping state-action pairs to non-negative rationals.
In practice, we often use these to represent ``costs'' (e.g. elapsed time or energy consumption),
despite the terminology ``rewards''.
The restriction to non-negative rewards allows us to avoid problems with non-uniqueness
of total rewards, which would require special treatment \cite{TV87}.

Rewards are accumulated along a path and,
for strategies
$\ctrlstrat\in\strats_\sgame^\Diamond$ and $\envstrat\in\strats_\sgame^\Box$,
the \emph{expected total reward} is defined as:
\[
E_{\sgame,\sinit}^{\ctrlstrat,\envstrat}(r) \rmdef\int_{\omega=s_0 a_0 s_1 a_1 \dots\in\ipaths_\sinit} \ \sum_{j=0}^{\infty} r(s_j,a_j) \,d\Prb_{\sgame,\sinit}^{\ctrlstrat,\envstrat}.
\]
For technical reasons, we will always assume the maximum possible reward
$\sup_{\sigma,\pi}E_{\sgame,s}^{\ctrlstrat,\envstrat}(r)$ is finite
(which can be checked with an analysis of the game's underlying graph); similar assumptions are commonly introduced~\cite[Section 7]{Put94}.
In our proofs, we will also use $E_{\sgame,\sinit}^{\ctrlstrat,\envstrat}(r{\downarrow}s)$ for the expected total reward accumulated before the first visit to $s$, defined by:
\[
 E_{\sgame,\sinit}^{\ctrlstrat,\envstrat}(r{\downarrow}s) \rmdef\int_{\omega=s_0 a_0 s_1 a_1\ldots\in\ipaths_\sinit}\sum_{i=0}^{\mathit{fst}(s,\omega)-1} r(s_i,a_i)\,d\Prb_{\sgame,\sinit}^{\ctrlstrat,\envstrat}
\]
where $\mathit{fst}(s,\omega)$ is $\min\{i \mid s_i = s\}$ if $\omega=s_0 a_0 s_1 a_1\ldots$ contains $s_i$, and $\infty$ otherwise.

An expected reward property is written $\phi = \propbt{r}{b}$ (where $\mathtt{C}$ stands for \emph{cumulative}),
meaning that the expected total reward for $r$ satisfies $\bowtie b$.
We say that $\phi$ is satisfied under strategies $\ctrlstrat$ and $\envstrat$,
denoted $\sgame,\ctrlstrat,\envstrat \models \phi$, if
$E_{\sgame,\sinit}^{\ctrlstrat,\envstrat}(r) \bowtie b$.

In fact, probabilistic reachability can easily be reduced to expected total reward
(by replacing any outgoing transitions from states in the target set
with a single transition to a sink state labelled with a reward of 1).
Thus, in the techniques presented in this paper,
we focus purely on expected total reward.


\subsection{Controller Synthesis}

To perform controller synthesis, we model the system as a stochastic game $\sgame=\gametuple$,
where player $\Diamond$ represents the controller and player $\Box$ represents the environment.
A specification of the required behaviour of the system is a property $\phi$,
either a probabilistic reachability property $\calP_{\bowtie p}[\,\future g\,]$
or an expected total reward property $\propbt{r}{b}$.

%
\begin{defi}[Sound strategy]\label{defn:soundstrat}
A strategy $\ctrlstrat\in\strats_\sgame^\Diamond$
for player $\Diamond$ in stochastic game $\sgame$
is \emph{sound} for a property $\phi$ if $\sgame,\ctrlstrat,\envstrat \models \phi$
for any strategy $\envstrat\in\strats_\sgame^\Box$.
\end{defi}
To simplify notation, we will consistently use $\ctrlstrat$ and $\envstrat$ to refer to strategies
of player $\Diamond$ and $\Box$, respectively,
and will not always state explicitly that
$\ctrlstrat\in\strats_\sgame^\Diamond$ and $\envstrat\in\strats_\sgame^\Box$.
Notice that, in \defref{defn:soundstrat}, strategies $\ctrlstrat$ and $\envstrat$ are both memoryless.
We could equivalently allow $\envstrat$ to range over history-dependent strategies
since, for the properties $\phi$ considered in this paper
(probabilistic reachability and expected total reward),
the existence of a history-dependent counter-strategy $\envstrat$ for which $\sgame,\ctrlstrat,\envstrat \not\models \phi$
implies the existence of a memoryless one.

The classical \emph{controller synthesis} problem asks whether there exists a sound strategy
for game $\sgame$ and property $\phi$. We can determine whether this is the case
by computing the optimal strategy for player $\Diamond$ in $\sgame$~\cite{Con93,FV97}.
This problem is known to be in NP $\cap$ co-NP, but, in practice,
methods such as value or policy iteration can be used efficiently.


\begin{figure}
  \centering
  \input{fig_running_game}
  \caption{A stochastic game $\sgame$ used as a running example (see \egref{eg:game}).}
  \label{fig:running_game}
\end{figure}


\begin{exa}\label{eg:game}
\figref{fig:running_game} shows a stochastic game $\sgame$, 
with controller and environment player states drawn as diamonds and squares, respectively.
It models the control of a robot moving between 4 locations ($s_0,s_2,s_3,s_5$).
When moving east ($s_0{\ra}s_2$ or $s_3{\ra}s_5$), it may be impeded by a second robot,
depending on the position of the latter.
If it is impeded, there is a chance that it does not successfully move to the next location.

We use a reward structure $\mathit{moves}$, which assigns 1 to the controller actions
$\mathit{north}$, $\mathit{east}$, $\mathit{south}$,
and define property $\phi=\nprop{\mathit{moves}}{5}$,
meaning that the expected number of moves to reach $s_5$ is at most 5
(notice that $s_5$ is the only state from which all subsequent transitions have reward zero).
A sound strategy for $\phi$ in $\sgame$ (found by minimising $\mathit{moves}$)
chooses $\mathit{south}$ in $s_0$ and $\mathit{east}$ in $s_3$,
yielding an expected number of moves of $3.5$.
\end{exa}

\section{Permissive Controller Synthesis} \label{sec:permissive_cont_synthesis}

We now define a framework for \emph{permissive controller synthesis},
which generalises classical controller synthesis by producing \emph{multi-strategies}
that offer the controller flexibility about which actions to take in each state.


\subsection{Multi-Strategies}\label{sec:multistrats}

Multi-strategies generalise the notion of strategies, as defined in \sectref{sec:prelims}.
They will always be defined for player $\Diamond$ of a game.
\begin{defi}[Multi-strategy]\label{defn:mstrat}
Let $\sgame=\gametuple$ be a stochastic game.
A (memoryless) \emph{multi-strategy} for $\sgame$
is a function $\mstrat:S_\Diamond{\rightarrow}\dist(2^{\act})$ with
$\mstrat(s)(\emptyset) = 0$ for all $s\in S_\Diamond$.
\end{defi}
As for strategies, a multi-strategy $\mstrat$ is
deterministic if $\mstrat$ always returns a Dirac distribution,
and randomised otherwise.
We write $\mstrats_\sgame^\dete$ and $\mstrats_\sgame^\rand$
for the sets of all deterministic and randomised multi-strategies in $\sgame$, respectively.

A deterministic multi-strategy $\mstrat$ chooses a set of \emph{allowed actions} in each state $s\in S_\Diamond$,
i.e., those in the unique set $B\subseteq\act$ for which $\mstrat(s)(B)=1$.
When $\mstrat$ is deterministic, we will often abuse notation
and write $a\in\mstrat(s)$ for the actions $a\in B$.
The remaining actions $\smash{\act(s)\setminus B}$ are said to be \emph{disallowed} in $s$.
In contrast to classical controller synthesis, where a strategy $\strat$
can be seen as providing instructions about precisely which action to take in each state,
in permissive controller synthesis a multi-strategy provides (allows) multiple actions, any of which can be taken.
A randomised multi-strategy generalises this by selecting a set of allowed actions in state $s$
randomly, according to distribution $\mstrat(s)$.

We say that a controller strategy $\ctrlstrat$ \emph{complies} with
multi-strategy $\mstrat$, denoted $\sigma \Compl \mstrat$, if it picks actions that are allowed by
$\mstrat$. Formally (taking into account the possibility of randomisation),
we define this as follows.
\begin{defi}[Compliant strategy]\label{def:compliant}
Let $\mstrat$ be a multi-strategy and $\ctrlstrat$ a strategy for a game $\sgame$.
We say that $\ctrlstrat$ is \emph{compliant} (or that it \emph{complies}) with $\mstrat$,
written $\sigma \Compl \mstrat$, if, for any state $s\in S_\Diamond$ and non-empty subset $B\subseteq A(s)$,
there is a distribution
$d_{s,B}\in\Dist(B)$ such that, for all $a\in \act(s)$, $\ctrlstrat(s)(a) = \sum_{B\ni a}\mstrat(s)(B)\cdot d_{s,B}(a)$.
\end{defi}


\begin{exa}
Let us explain the technical definition of a compliant strategy on the game
from \egref{eg:game} (see \figref{fig:running_game}).
Consider a randomised multi-strategy $\mstrat$ that, in $s_0$, picks $\{\mathit{east}, \mathit{south}\}$ with probability
$0.5$, $\{\mathit{south}\}$ with probability $0.3$, and $\{\mathit{east}\}$ with probability $0.2$.
A compliant strategy then needs to, for some number $0\le x\le 1$,
pick $\mathit{south}$ with probability $0.3 + 0.5\cdot x$ and $\mathit{east}$ with probability $0.2 + 0.5\cdot (1-x)$.
The number $x$ corresponds to the probability
$d_{s_0,\{\mathit{east}, \mathit{south}\}}(\mathit{south})$ in the formal definition above.

Hence, a strategy $\ctrlstrat$ that picks $\mathit{east}$ and $\mathit{south}$
with equal probability 0.5 satisfies the requirements
of compliance in state $s_0$, as witnessed by selecting $x=0.4$, or, in other words,
the distribution $d_{s_0,\{\mathit{east}, \mathit{south}\}}$
assigning $0.4$ and $0.6$ to $\mathit{south}$ and $\mathit{east}$, respectively. On the other hand,
a strategy that picks $\mathit{east}$ with probability $0.8$ cannot be compliant with $\mstrat$.
\end{exa}

Each multi-strategy determines a set of compliant strategies, and our aim is to
design multi-strategies which allow as many actions as possible, but at the same
time ensure that any compliant strategy satisfies some specified property.
We define the notion of a \emph{sound} multi-strategy,
i.e., one that is guaranteed to satisfy a property $\phi$ when complied with.

\begin{defi}[Sound multi-strategy]\label{defn:soundmstrat}
A multi-strategy $\mstrat$ for game $\sgame$ is \emph{sound} for a property $\phi$
if any strategy $\ctrlstrat$ that complies with $\mstrat$ is sound for $\phi$.
\end{defi}



\begin{exa}\label{eg:mstrat}
We return again to the stochastic game
from \egref{eg:game} (see \figref{fig:running_game}) and re-use the property
$\phi=\nprop{\mathit{moves}}{5}$.
A strategy that picks $\mathit{south}$ in $s_0$ and $\mathit{east}$ in $s_3$
results in an expected reward of 3.5 (i.e., 3.5 moves on average to reach $s_5$).
A strategy that picks $\mathit{east}$ in $s_0$ and $\mathit{south}$ in $s_2$
yields expected reward 5.
Thus, a (deterministic) \emph{multi-strategy} $\mstrat$ that picks $\{\mathit{south},\mathit{east}\}$ in $s_0$,
$\{\mathit{south}\}$ in $s_2$ and $\{\mathit{east}\}$ in $s_3$ is sound for $\phi$
since the expected reward is always at most 5.
\end{exa}


\subsection{Penalties and Permissivity}\label{sec:penalties}

The motivation behind synthesising multi-strategies is to offer flexibility
in the actions to be taken, while still satisfying a particular property $\phi$.
Generally, we want a multi-strategy $\mstrat$ to be as \emph{permissive} as possible,
i.e., to impose as few restrictions as possible on actions to be taken.
We formalise the notion of permissivity by assigning \emph{penalties} to actions in the model,
which we then use to quantify the extent to which actions are disallowed by $\mstrat$.
Penalties provide expressivity in the way that we quantify permissivity:
if it is more preferable that certain actions are allowed than others,
then these can be assigned higher penalty values.

A \emph{penalty scheme} is a pair $(\penbase,\pentype)$, comprising
a \emph{penalty function} $\penbase:S_\Diamond\times\act\to\Qset_{\ge 0}$
and a \emph{penalty type} $\tau\in\{\sta,\dyn\}$.
The function $\penbase$ represents the impact of disallowing each action
in each controller state of the game. 
The type $\pentype$ dictates how penalties for individual actions are
combined to quantify the permissivity of a specific multi-strategy.
For \emph{static penalties} ($\pentype=\sta$),
we simply sum penalties across all states of the model.
For \emph{dynamic penalties} ($\pentype=\dyn$),
we take into account the likelihood that disallowed actions would actually have been available,
by using the \emph{expected sum} of penalty values.

More precisely, for a penalty scheme $(\penbase,\pentype)$ and a multi-strategy $\mstrat$,
we define the resulting penalty for $\mstrat$, denoted $\pendef{\penbase}{\pentype}{\mstrat}$ as follows.
First, we define the \emph{local} penalty for $\mstrat$ at state $s\in S_\Diamond$ as
%
%
$\pendef{\penbase}{loc}{\mstrat,s} = \!\sum_{B\subseteq\act(s)}\!\sum_{a\notin B}\!\mstrat(s)(B)\penbase(s,a)$.
If $\mstrat$ is deterministic, $\pendef{\penbase}{loc}{\mstrat,s}$
is simply the sum of the penalties of actions that are disallowed by $\mstrat$ in $s$.
If $\mstrat$ is randomised, $\pendef{\penbase}{loc}{\mstrat,s}$ gives the expected penalty value in $s$,
i.e., the sum of penalties weighted by the probability with which $\mstrat$ disallows them in $s$.

Now, for the static case, we sum the local penalties over all states, i.e., we put:
$$\pendef{\penbase}{\sta}{\mstrat} = \sum\nolimits_{s\in S_\Diamond}\pendef{\penbase}{loc}{\mstrat,s}.$$
For the dynamic case, we use the (worst-case) expected sum of local penalties.
We define an auxiliary reward structure $\penrew{\mstrat}$
given by the local penalties: $\penrew{\mstrat}(s,a) = \pendef{\penbase}{loc}{\mstrat,s}$ for all $s\in S_\Diamond$ and $a\in\act(s)$,
and $\penrew{\mstrat}(s,a) = 0$ for all $s\in S_\Box$ and $a\in \act(s)$.
Then:
$$\smash{\pendef{\penbase}{\dyn}{\mstrat,s}
= \sup \{ E_{\sgame,s}^{\ctrlstrat,\envstrat}(\penrew{\mstrat})} \, | \, \ctrlstrat\in\strats_\sgame^\Diamond,
\envstrat\in\strats_\sgame^\Box \mbox{ and $\sigma$ complies with $\mstrat$} \}.$$
We use $\pendef{\penbase}{\dyn}{\mstrat} = \pendef{\penbase}{\dyn}{\mstrat,\sinit}$
to reference the dynamic penalty in the initial state.

%



\subsection{Permissive Controller Synthesis}
We can now formally define the central problem studied in this paper.

\begin{defi}[Permissive controller synthesis]\label{defn:pcs}
Consider a game $\sgame$,
a class of multi-strategies $\star\in\{\dete,\rand\}$,
a property $\phi$, 
a penalty scheme $(\penbase,\pentype)$
and a threshold $c\in\Qset_{\geq0}$.
The \emph{permissive controller synthesis} problem asks:
does there exist a multi-strategy $\mstrat\in\mstrats_\sgame^\star$ that is sound for $\phi$
and satisfies $\pendef{\penbase}{\pentype}{\mstrat} \le c$?
\end{defi}
Alternatively, in a more quantitative fashion, we can aim to synthesise (if it exists)
an \emph{optimally permissive} sound multi-strategy.

\begin{defi}[Optimally permissive]\label{defn:opt}
Let $\sgame$, $\star$, $\phi$ and $(\penbase,\pentype)$ be as in \defref{defn:pcs}.
A sound multi-strategy $\hat{\mstrat}\in\mstrats_\sgame^\star$ is \emph{optimally permissive}
if its penalty $\pendef{\penbase}{\pentype}{\hat{\mstrat}}$
equals the infimum $\inf\{ \pendef{\penbase}{\pentype}{\mstrat}  \,|\,\mstrat\in\mstrats_\sgame^{\star}\mbox{ and $\mstrat$ is sound for $\phi$}\}$.
\end{defi}


\begin{exa}\label{eg:pcs}
We return to \egref{eg:mstrat} and consider a static penalty scheme $(\penbase,\sta)$
assigning 1 to the actions $\mathit{north}$, $\mathit{east}$, $\mathit{south}$ (in any state).
The deterministic multi-strategy $\mstrat$ from \egref{eg:mstrat} is optimally permissive
for $\phi=\nprop{\mathit{moves}}{5}$, with 
penalty 1 (just $\emph{north}$ in $s_3$ is disallowed).
If we instead use $\phi'=\nprop{\mathit{moves}}{16}$,
the multi-strategy $\mstrat'$ that extends $\mstrat$ by also allowing $\mathit{north}$
is now sound and optimally permissive, with 
penalty 0.
Alternatively, the randomised multi-strategy $\mstrat''$ that picks
$\{\mathit{north}\}$ with probability 0.7 and $\{\mathit{north},\mathit{east}\}$ with probability 0.3 in $s_3$
is sound for $\phi$ with penalty just 0.7.
\end{exa}

It is important to point out that penalties will typically be used for {\em relative comparisons} of multi-strategies.
If two multi-strategies $\mstrat$ and $\mstrat'$ incur penalties $x$ and $x'$ with $x<x'$,
then the interpretation is that $\mstrat$ is better than $\mstrat'$; there is not necessarily any intuitive
meaning assigned to the values $x$ and $x'$ themselves. Accordingly, when modelling a system, the penalties of actions
should be chosen to reflect the actions' relative importance.
This is different from rewards, which usually correspond to a specific measure of the system.

Next, we establish several fundamental results
about the permissive controller synthesis problem.
Proofs that are particularly technical are postponed to the appendix
and we only highlight the key ideas in the main body of the paper.


\startpara{Optimality}
Recall that two key parameters of the problem are
the type 
of multi-strategy sought (deterministic or randomised)
and the type of 
penalty scheme used (static or dynamic).
%
%
We first note that \emph{randomised} multi-strategies
are strictly more powerful than deterministic ones,
i.e., they can be more permissive (yield a lower penalty)
whilst satisfying the same property $\phi$.
\begin{thm}\label{thm:rand}
The answer to a permissive controller synthesis problem
(for either a \emph{static} or \emph{dynamic} penalty scheme)
can be ``no'' for \emph{deterministic} multi-strategies,
but ``yes'' for \emph{randomised} ones.
\end{thm}

%
%
%
\startpara{Proof}
Consider an MDP with states $s$, $t_1$ and $t_2$, and actions $a_1$ and $a_2$, where
$\delta(s,a_i)(t_i)=1$ for $i\in \{1,2\}$, and $t_1,t_2$ have self-loops only.
Let $r$ be a reward structure assigning $1$ to $(s,a_1)$
and $0$ to all other state-action pairs,
and $\penbase$ be a penalty function
assigning $1$ to $(s,a_2)$ and $0$ elsewhere.
We then ask whether there is a multi-strategy satisfying $\phi=\property{r}{0.5}$
with penalty at most 0.5.

Considering either static or dynamic penalties,
the randomised multi-strategy $\mstrat$ that chooses distribution $0.5{:}\{a_1\}+0.5{:}\{a_2\}$ in $s$
is sound and yields penalty 0.5.
However, there is no such deterministic multi-strategy.
\qed

This is why we explicitly distinguish between classes of multi-strategies
when defining permissive controller synthesis.
This situation contrasts with classical controller synthesis,
where deterministic strategies are optimal for the same classes of properties $\phi$.
Intuitively, randomisation is more powerful in this case because
of the trade-off between rewards and penalties:
similar results exist in, for example,
multi-objective controller synthesis on MDPs~\cite{EKVY08}.

Next, we observe that, for the case of static penalties,
the optimal penalty value for a given property (the infimum of achievable values)
may not actually be achievable by any randomised multi-strategy.
\begin{thm}\label{thm:opt}
For permissive controller synthesis using a \emph{static} penalty scheme,
an optimally permissive \emph{randomised} multi-strategy does not always exist.
\end{thm}
%
%
%
%
\startpara{Proof}
Consider a game with states $s$ and $t$, and actions $a$ and $b$, where we define
$\delta(s,a)(s)=1$ and $\delta(s,b)(t)=1$, and $t$ has just a self-loop.
The reward structure $r$ assigns $1$ to $(s,b)$ and $0$ to all other state-action pairs.
The penalty function $\penbase$ assigns $1$ to $(s,a)$ and $0$ elsewhere.

Now observe that any multi-strategy which disallows the action $a$ with probability $\varepsilon>0$ and allows all other actions incurs penalty $\varepsilon$ and is sound for $\property{r}{1}$, since any strategy which complies with the multi-strategy leads to action $b$ being taken eventually. Thus, the infimum of achievable penalties is $0$. However, the multi-strategy that incurs penalty $0$, i.e. allows all actions, is not sound for $\property{r}{1}$.
\qed

If, on the other hand, we restrict our attention to deterministic strategies,
then an optimally permissive multi-strategy \emph{does} always exist
(since the set of deterministic, memoryless multi-strategies is finite).
For randomised multi-strategies with dynamic penalties, the question remains open.


\startpara{Complexity}
Next, we present complexity results for the different variants of the
permissive controller synthesis problem.
We begin with lower bounds.


\begin{thm}\label{hardness}
The permissive controller synthesis problem is NP-hard,
for either \emph{static} or \emph{dynamic} penalties,
and {\em deterministic} or {\em randomised} multi-strategies.
\end{thm}
We prove NP-hardness by reduction from the Knapsack problem,
where weights of items are represented by penalties, and
their values are expressed in terms of rewards to be achieved.
The most delicate part is the proof for randomised strategies, where we need to ensure that the multi-strategy cannot benefit from picking certain actions (corresponding to items being put into the Knapsack) with probability other than  $0$ or $1$.
See \appref{appx:md_hardness} for details.
For upper bounds, we have the following.
\begin{thm}\label{upper-bound}
The permissive controller synthesis problem for \emph{deterministic} (resp. {\em randomised}) strategies is in NP (resp.\ PSPACE) for {\em dynamic}/\emph{static} penalties.
\end{thm}
For deterministic multi-strategies, it is straightforward to show NP membership in both the dynamic and static penalty case, since we can guess a multi-strategy satisfying the required conditions
and check its correctness in polynomial time.
For randomised multi-strategies, with some technical effort,
we can encode existence of the required multi-strategy as a formula of the existential fragment of the theory of real arithmetic,
solvable with polynomial space~\cite{Canny}.
See \appref{appx:upper-bound}.
A natural question is whether the PSPACE upper bound for randomised multi-strategies can be improved.
We show that this is likely to be difficult, by giving a reduction from the square-root-sum problem. 

\begin{thm}\label{square-root-sum}
 There is a reduction from the square-root-sum problem to the permissive controller synthesis problem with {\em randomised} multi-strategies, for both {\em static} and {\em dynamic} penalties.
\end{thm}

We use a variant of the problem that 
asks, given positive rationals $x_1$,\dots,$x_n$ and $y$,
whether $\sum_{i = 1}^n \sqrt{x_i} \le y$.
This problem is known to be in PSPACE, but establishing a better complexity bound is a long-standing open problem in computational geometry \cite{Garey}. 
See \appref{appx:reduction} for details.

\section{MILP-Based Synthesis of Multi-Strategies}\label{sec:synthesis}

We now consider practical methods for synthesising
multi-strategies that are sound for a property $\phi$
and optimally permissive for some penalty scheme.
Our methods use mixed integer linear programming (MILP),
which optimises an objective function
subject to linear constraints that mix both real and integer variables.
A variety of efficient, off-the-shelf MILP solvers exists.

An important feature of the MILP solvers we use is that they work incrementally,
producing a sequence of increasingly good solutions.
Here, that means generating
a series of sound multi-strategies that are increasingly permissive.
In practice, when computational resources are constrained, it may be acceptable
to stop early and accept a multi-strategy that is sound but not necessarily optimally permissive.

Here, and in the rest of this section, we assume that the property $\phi$
is of the form $\property{r}{b}$. Upper bounds on expected rewards ($\phi=\nprop{r}{b}$)
can be handled by negating rewards and converting to a lower bound.
%
%
For the purposes of encoding into MILP,
we rescale $r$ and $b$ such that
$\sup_{\sigma,\pi}E_{\sgame,s}^{\sigma,\pi}(r)<1$ for all $s$,
and rescale every (non-zero) penalty
such that $\penbase(s,a) \geq 1$ for all $s$ and $a \in A(s)$.

We begin by discussing the synthesis of deterministic multi-strategies,
first for static penalties and then for dynamic penalties.
Subsequently, we present an approach to synthesising approximations
to optimal randomised multi-strategies.
In each case, we describe encodings into MILP problems and prove their correctness.
We conclude this section with a brief discussion of ways to optimise the MILP encodings.
Then, in \sectref{sec:expts}, we investigate the practical applicability of our techniques.


\subsection{Deterministic Multi-Strategies with Static Penalties}\label{sec:detsynthsta}

\begin{figure}[!t]
\begin{flalign*}
\mbox{Minimise:} & \
- x_{\sinit}\  + \sum\nolimits_{s \in S_{\Diamond}} \sum\nolimits_{a \in A(s)} (1-y_{s,a}){\cdot}\penbase(s,a) \ \ \mbox{ subject to:} &
\end{flalign*}
\begin{align}
\label{eqn:det-fst}x_{\sinit} &\geq b\\
\label{eqn:det-mstrat}1&\le \sum\nolimits_{a \in {A(s)}} y_{s,a} &\text{for all $s\in S_\Diamond$}\\
\label{eqn:det-rew1}x_{s} &\leq \sum\nolimits_{t \in S} \delta(s,a)(t){\cdot} x_t + r(s,a) + (1-y_{s,a}) &\text{for all $s\in S_\Diamond$, $a\in \act(s)$}\\
\label{eqn:det-rew2}x_{s} &\leq \sum\nolimits_{t \in S} \delta(s,a)(t){\cdot} x_{t} + r(s,a) &\text{for all $s\in S_\Box$, $a\in \act(s)$}\\
x_s & \le \alpha_s &\text{for all $s\in S$}\label{eq:path-a}\\
y_{s,a} & = (1-\alpha_s) + \!\!\sum\nolimits_{t\in \support{\delta(s,a)}} \beta_{s,a,t} & \text{for all $s\in S, a\in \act(s)$}\label{eq:path-b}\\
y_{s,a} & = 1 & \text{for all $s\in S_\Box , a\in \act(s)$}\label{eq:path-b2}\\[0.4em]
\gamma_t & < \gamma_s + (1 - \beta_{s,a,t}) + r(s,a) & \text{for all $s$, $a$, $t$ with $t\in\support{\delta(s,a)}$}\label{eq:path-c}
\end{align}
\vspace*{-1em}
\caption{MILP encoding for deterministic multi-strategies with static penalties.}
\label{lp_encoding_static}
\end{figure}
%
\begin{figure}[!t]
Minimise: $z_{\sinit}\ $ subject to (\ref{eqn:det-fst}),\dots,(\ref{eq:path-c}) and:
\begin{align}
\ell_{s} &= \sum\nolimits_{a\in\act(s)}\penbase(s,a){\cdot}(1-y_{s,a}) &\text{for all $s\in S_\Diamond$}\label{eq:dyn-diamond-a}\\
z_{s} &\geq \sum\nolimits_{t \in S} \delta(s,a)(t){\cdot} z_{t} + \ell_{s} - c{\cdot}(1-y_{s,a}) &\text{for all $s\in S_\Diamond$, $a\in \act(s)$}\label{eq:dyn-diamond-b}\\
z_{s} &\geq \sum\nolimits_{t \in S} \delta(s,a)(t){\cdot} z_{t} &\text{for all $s\in S_\Box$, $a\in \act(s)$}\label{eq:dyn-box}
\end{align}
\vspace*{-1em}
\caption{MILP encoding for deterministic multi-strategies with dynamic penalties.}
\label{lp_encoding_dynamic}
\end{figure}

\figref{lp_encoding_static} shows an encoding into MILP
of the problem of finding an optimally permissive \emph{deterministic} multi-strategy
for property $\phi=\property{r}{b}$ and a \emph{static} penalty scheme $(\penbase,\sta)$.
%
%
The encoding uses 5 types of variables:
$y_{s,a} \in \{0,1\}$,
$x_s \in [0,1]$,
$\alpha_s\in \{0,1\}$, $\beta_{s,a,t}\in \{0,1\}$ and $\gamma_t\in [0,1]$,
where $s,t\in S$ and $a\in\act$.
The worst-case size of the MILP problem is $\mathcal{O}(|A|{\cdot}|S|^2{\cdot}\kappa)$,
where $\kappa$ stands for the longest encoding of a number used.

Variables $y_{s,a}$ encode a multi-strategy $\mstrat$ as follows:
$y_{s,a}$ has value $1$ iff $\mstrat$ allows action $a$ in $s\in S_\Diamond$
(constraint \eqref{eqn:det-mstrat} enforces at least one allowed action per state).
Variables $x_s$ represent the worst-case expected total reward (for $r$) from state $s$,
under any controller strategy complying with $\mstrat$ and under any environment strategy.
This is captured by constraints \eqref{eqn:det-rew1}--\eqref{eqn:det-rew2}
(which are analogous to the linear constraints used when minimising the reward in an MDP).
Constraint \eqref{eqn:det-fst} puts the required bound of $b$ on the reward from $\sinit$.

The objective function minimises the static penalty (the sum of all local penalties)
minus the expected reward in the initial state.
The latter acts as a tie-breaker between solutions with equal penalties
(but, thanks to rescaling, is always dominated by the penalties
and therefore does not affect optimality).

As an additional technicality, we need to ensure
the values of $x_s$ are the \emph{least} solution of the defining inequalities,
to deal with the possibility of zero reward loops.
To achieve this, we use an approach similar to the one taken in~\cite{WJABK12}.
%
It is sufficient to ensure that $x_s=0$
whenever the minimum expected reward from $s$ achievable under $\mstrat$ is $0$,
which is true if and only if, starting from $s$, it is possible to avoid ever taking an action with positive reward.

In our encoding, $\alpha_s=1$  if $x_s$ is positive (constraint~\eqref{eq:path-a}). The binary variables $\beta_{s,a,t}=1$ represent, for each such $s$ and each action $a$ allowed in $s$, a choice of successor $t=t(s,a) \in\support{\delta(s,a)}$ (constraint~\eqref{eq:path-b}). The variables $\gamma_s$ then represent a ranking function: if $r(s,a)=0$, then $\gamma_s>\gamma_{t(s,a)}$ (constraint~\eqref{eq:path-c}). If a positive reward could be avoided starting from $s$, there would in particular be an infinite sequence $s_0,a_1,s_1,\dots$ with $s_0=s$ and, for all $i$, either (i) $x_{s_i} > x_{s_{i+1}}$, or (ii) $x_{s_i} = x_{s_{i+1}}$, $s_{i+1}=t(s_i,a_i)$ and $r(s_i,a_i)=0$, and therefore $\gamma_{s_i}>\gamma_{s_{i+1}}$.
This means that the sequence $(x_{s_0}, \gamma_{s_0}),(x_{s_1}, \gamma_{s_1}),\dots$
is (strictly) decreasing w.r.t. the lexicographical order, but at the same time $S$ is finite, and so this sequence would have to enter a loop, which is a contradiction.

\startpara{Correctness}
Before proving the correctness of the encoding (stated in \thmref{thm:main-static}, below),
we prove the following auxiliary lemma that characterises the reward achieved
under a multi-strategy in terms of a solution of a set of inequalities.

\begin{lem}\label{lem:value-soln}
Let $\sgame=\gametuple$ be a stochastic game, $\phi=\property{r}{b}$ a property,
$(\penbase,\sta)$ a \emph{static} penalty scheme and $\mstrat$ a deterministic multi-strategy.
Consider the inequalities: 
$$\begin{array}{lcll}
x_s & \le & \min_{a\in \mstrat(s)} \sum_{s' \in S} \delta(s,a)(s') x_{s'} + r(s,a) & \mbox{ for $s\in S_\Diamond$} \\
x_s & \le & \min_{a\in \act(s)} \sum_{s' \in S} \delta(s,a)(s') x_{s'} + r(s,a) & \mbox{ for $s\in S_\Box$} .
\end{array}$$
Then the following hold:
\begin{itemize}
 \item $\bar x_s = \inf_{\sigma \Compl \mstrat, \pi}E_{\sgame,s}^{\sigma,\pi}(r)$ is a solution to the above inequalities.
 \item A solution $\bar x_s$ to the above inequalities satisfies
  $\bar x_s\le \inf_{\sigma\Compl\mstrat, \pi}E_{\sgame,s}^{\sigma,\pi}(r)$ for all $s$ whenever the following condition holds:
  for every $s$ with $\bar x_s>0$, 
  every $\sigma \Compl\mstrat$ and every $\pi$
  there is a path $\omega = s_0a_0\ldots s_n a_n$
  starting in $s$ that satisfies $\Prb_{\sgame,s}^{\sigma,\pi}(\omega)>0$
  and $r(s_n,a_n) > 0$.
\end{itemize}
\end{lem}
\begin{proof}
The game $\sgame$, together with $\mstrat$, determines a Markov decision process $\sgamestrat= \langle \emptyset, S_\Diamond \cup S_\Box, \sinit, \act, \delta'\rangle$ in which the choices disallowed by $\mstrat$ are removed, i.e.
$\delta'(s,a)$ is equal to $\delta(s,a)$ for every $s\in S_\Box$ and every $s \in S_\Diamond$ with $a \in \mstrat(s)$, and is undefined for any other combination of $s$ and $a$.
We have:
\[
 \inf_{\sigma \Compl\mstrat, \pi}E_{\sgame,s}^{\sigma,\pi}(r)
 = \inf_{\bar\sigma}E_{\sgame^\mstrat,s}^{\bar\sigma}(r)
 \]
since, for any strategy pair $\sigma \Compl\mstrat$ and $\pi$ in $\sgame$,
there is a strategy $\bar\sigma$ in $\sgamestrat$ which is defined,
for every finite path $\omega$ of $\sgamestrat$ ending in $t$,
by $\bar\sigma(\omega) = \sigma(\omega)$ or $\bar\sigma(\omega) = \pi(\omega)$,
depending on whether $t\in S_\Diamond$ or $t \in S_\Box$, and which satisfies
$\smash{E_{\sgame,s}^{\sigma,\pi}(r) = E_{\sgame^\mstrat,s}^{\bar\sigma}(r)}$.
Similarly, a strategy $\bar\sigma$ for $\sgamestrat$ induces
a compliant strategy $\sigma$ and a strategy $\pi$
defined for every finite path $\omega$ of $\sgame$ ending in $S_\Diamond$
(resp.\ $S_\Box$) by $\sigma(\omega) = \bar\sigma(\omega)$
(resp.\ $\pi(\omega) = \bar\sigma(\omega)$).

The rest is then the following simple application of results from the theory of Markov decision processes.
The first item of the lemma follows from~\cite[Theorem~7.1.3]{Put94}, which gives a characterisation of values in MDPs in terms of Bellman equations; the inequalities in the lemma are in fact a relaxation of these equations.
For the second part of the lemma, observe that if,
$\inf_{\sigma\Compl\mstrat,\pi}E_{\sgame,s}^{\sigma,\pi}(r)$ is infinite, then the claim holds trivially.
Otherwise, from the assumption on the existence of
$\omega$ we have that, under any compliant strategy,
there is a path $\omega' = s_0 a_0 s_1 \ldots s_n$ of length at most $|S|$ in
$\sgamestrat$ such that $\inf_{\sigma \Compl\mstrat,\pi}E_{\sgame,s_n}^{\sigma,\pi}(r) = 0$ (otherwise the reward would be infinite) and so $\bar x_{s_n} = 0$. We can thus apply~\cite[Proposition~7.3.4]{Put94},
which states that a solution to our inequalities gives optimal values whenever under any strategy the probability of reaching a state $s$ with $x_s=0$ is $1$. Note that the result of~\cite{Put94} applies for maximisation of
reward in ``negative models''; our problem can be easily reduced to this setting by multiplying the rewards by $-1$ and looking for maximising (instead of minimising) strategies.
\end{proof}

\begin{thm}\label{thm:main-static}
Let $\sgame$ be a game, $\phi=\property{r}{b}$ a property and $(\penbase,\sta)$ a \emph{static} penalty scheme. There is a sound multi-strategy in $\sgame$ for $\phi$ with penalty $p$ if and only if there is an optimal assignment to the
MILP instance from~\figref{lp_encoding_static} which satisfies
$p=\sum\nolimits_{s \in S_{\Diamond}} \sum\nolimits_{a \in A(s)} (1-y_{s,a}){\cdot}\penbase(s,a)$.
\end{thm}
\begin{proof}
We prove that every multi-strategy $\mstrat$ induces a satisfying assignment to the variables such that the static penalty under $\mstrat$ is
$\sum\nolimits_{s \in S_{\Diamond}} \sum\nolimits_{a \in A(s)} (1-y_{s,a}){\cdot}\penbase(s,a)$, and vice versa.
The theorem then follows from the rescaling of rewards and penalties that we performed.

We start by proving that, given a sound multi-strategy $\mstrat$,
we can construct a satisfying assignment
$\{\bar y_{s,a},\bar x_s,\bar \alpha_s,\bar \beta_{s,a,t},\bar \gamma_t\}_{s,t\in S,a\in\act}$
to the constraints from~\figref{lp_encoding_static}.
For $s\in S_\Diamond$ and $a\in A(s)$ we set $\bar y_{s,a}=1$ if $a \in \mstrat(s)$, and otherwise we set
$\bar y_{s,a}=0$. This gives satisfaction of contraint~\eqnref{eqn:det-mstrat}.
For $s\in S_\Box$ and $a\in A(s)$ we set $\bar y_{s,a}=1$, ensuring satisfaction of \eqnref{eq:path-b2}.
We then put $\bar x_s = \inf_{\sigma\Compl\mstrat, \pi}E_{\sgame,s}^{\sigma,\pi}(r)$.
By the first part of \lemref{lem:value-soln} we get that constraints~\eqnref{eqn:det-fst},
\eqnref{eqn:det-rew1} (for $a\in \mstrat(s)$) and \eqnref{eqn:det-rew2} are satisfied.
Constraint~\eqnref{eqn:det-rew1} for $a\notin\mstrat(s)$ is satisfied because
in this case $\bar y_{s,a} = 0$, and so the right-hand side is always at least $1$.

We further set $\bar\alpha_s = 1$ if $x_s > 0$ and $\bar\alpha_s = 0$ if $x_s=0$, thus satisfying constraint~\eqnref{eq:path-a}. For a state $s$, let $d_s$ be the maximum distance
to a positive reward. Formally, the values $d_s$ are defined inductively by putting $d_s=0$ for any state $s$ such that
we have $r(s,a) > 0$ for all $a\in \enab{s}$, and then for any other state $s$:
\begin{align*}
 d_s &= 1 + \min_{a\in\mstrat(s),r(s,a)=0}\ \max_{\delta(s,a)(t) > 0}d_t\quad\text{if $s\in S_\Diamond$}\\
 d_s &= 1 + \min_{a\in\act(s),r(s,a)=0}\ \max_{\delta(s,a)(t) > 0}d_t\quad\text{if $s\in S_\Box$}
\end{align*}
Put $d_s = \bot$ if $d_s$ was not defined by the above.
For $s$ such that $d_s \neq \bot$, we put
$\bar\gamma_s = d_s/|S|$, and for every $a$ we choose
$t$ such that $d_t<d_s$, and set $\bar\beta_{s,a,t} = 1$, leaving $\bar\beta_{s,a,t}=0$ for all other $t$.
For $s$ such that $d_s = \bot$ we define $\bar\gamma_s=0$ and for all $a$ and
$t$ put $\bar\beta_{s,a,t} = 0$. This ensures the satisfaction of the remaining constraints.

In the opposite direction, assume that we are given a satisfying assignment.
Firstly, we create a game $\sgame'$ from $\sgame$ by making
any states $s$ with $\bar x_s=0$ sink states (i.e. imposing a self-loop with no penalty on $s$ and removing all other transitions). Any sound multi-strategy
$\mstrat$ for $\phi$ in $\sgame'$ directly gives a sound multi-strategy $\mstrat'$ for $\phi$ in $\sgame$ defined by $\mstrat'(s) = \mstrat(s)$ for states $s\in S_\Diamond$ with $x_s>0$, and otherwise letting $\mstrat$ allow all available actions.

We construct $\mstrat$ for $\sgame'$ by putting
$\mstrat(s) = \{a \in A(s) \mid \bar y_{s,a} = 1\}$
for all $s \in S_\Diamond$ with $\bar x_s > 0$, and by allowing the self-loop in the states $s\in S_\Diamond$ with $x_s = 0$; note that $\mstrat(s)$ is non-empty by constraint~\eqnref{eqn:det-mstrat}.
First, by definition, the multi-strategy yields the penalty $\sum\nolimits_{s \in S_{\Diamond}} \sum\nolimits_{a \in A(s)} (1-\bar y_{s,a}){\cdot}\penbase(s,a)$.
Next, we will show that $\mstrat$ satisfies the assumption of
the second part of \lemref{lem:value-soln},
from which we get that:
\[
 \inf_{\sigma\Compl\mstrat, \pi}E_{\sgame',s}^{\sigma,\pi}(r)
 \ge \bar x_s
\]
which, together with constraint~\eqnref{eqn:det-fst} being satisfied, gives us the desired result.

Consider any $s$ such that
$\inf_{\sigma\Compl\mstrat, \pi}E_{\sgame',s}^{\sigma,\pi}(r) > 0$.
Then we have $\bar x_s > 0$ (by the definition of $\sgame'$). Let us fix any $\sigma\Compl\mstrat$ and any $\pi$,
and let $s_0 = s$.
We show that there is a path
$\omega$ satisfying the assumption of the lemma. We build $\omega = s_0\ldots s_n a_n$ inductively,
to satisfy: (i) $r(s_n,a_n) >0$,  (ii)
$\bar x_{s_i} \ge \bar x_{s_{i-1}}$ for all $i$, and (iii) for any sub-path $s_i a_i \ldots s_j$ with
$\bar x_{s_i} = \bar x_{s_j}$ we have that $\bar\gamma_{s_k} < \bar\gamma_{s_{k-1}}$ for all $i+1\le k \le j$.

Assume we have defined a prefix $s_0a_0\ldots s_i$ to satisfy conditions (ii) and (iii). We put $a_i$
to be the action picked by $\sigma$ (or $\pi$) in $s_i$.
If $r(s_i,a_i) > 0$, we are done.
Otherwise, we pick $s_{i+1}$ as follows:
\begin{itemize}
 \item If there is $s'\in\support{\delta(s_{i},a_i)}$ with $\bar x_{s'} > \bar x_s$, then we put $s_{i+1} = s'$.
  Such a choice again satisfies (ii) and (iii) by definition.
 \item If we have $\bar x_{s'} = \bar x_s$ for all $s'\in\support{\delta(s_{i},a_i)}$, then any choice will satisfy (ii). To satisfy the other conditions, we pick $s_{i+1}$ so that $\bar\beta_{s_i,a_i,s_{i+1}} = 1$
       is true. We argue that such an $s_{i+1}$ can be chosen. We have $\bar x_{s_i} > 0$ and so $\bar\alpha_s = 1$ by
       constraint~\eqnref{eq:path-a}. We also have $\bar y_{s,a} = 1$:
       for $s\in S_\Diamond$ this follows from the definition of $\mstrat$, for $s\in S_\Box$ from constraint~\eqnref{eq:path-b2}.
       Hence, since constraint~\eqnref{eq:path-b} is satisfied, there must be $s_{i+1}$ such that $\bar\beta_{s_i,a,s_{i+1}} = 1$.
       Then, we apply constraint~\eqnref{eq:path-c} (for $s=s_i$, $t=s_{i+1}$ and $a=a_i$) and,
       since the last two summands on the right-hand side
       are $0$, we get $\bar\gamma_{s_{i+1}} < \bar\gamma_{s_i}$, thus satisfying (iii).
\end{itemize}
Note that the above construction must terminate after at most $|S|$ steps since,
due to conditions (ii) and (iii), no state repeats on $\omega$.
Because the only way of terminating is satisfaction of (i), we are done.
\end{proof}

\subsection{Deterministic Multi-Strategies with Dynamic Penalties}\label{sec:detsynthdyn}
Next, we show how to compute a sound and optimally permissive \emph{deterministic} multi-strategy
for a \emph{dynamic} penalty scheme $(\penbase,\dyn)$.
This case is more subtle since the optimal penalty can be infinite.
Hence, our solution proceeds in two steps as follows.

Initially, we determine if there is {\em some} sound multi-strategy.
For this, we just need to check for the existence of a sound strategy,
using standard algorithms for solution of stochastic games~\cite{Con93,FV97}.
If there is no sound multi-strategy, we are done.
Otherwise,
we use the MILP problem in \figref{lp_encoding_dynamic}
to determine the penalty for an optimally permissive sound multi-strategy.
This MILP encoding extends the one in~\figref{lp_encoding_static} for static penalties,
adding variables $\ell_s$ and $z_s$,
representing the local and the expected penalty in state $s$,
and three extra sets of constraints.
First, \eqnref{eq:dyn-diamond-a} and \eqnref{eq:dyn-diamond-b} define the expected penalty in controller states,
which is the sum of penalties for all disabled actions and those in the successor states,
multiplied by their transition probabilities.
The behaviour of environment states is then captured by
constraint~\eqnref{eq:dyn-box}, where we only maximise the penalty,
without incurring any penalty locally.

The constant $c$ in \eqnref{eq:dyn-diamond-b} is chosen to be no lower than any
{\em finite} penalty achievable by a deterministic multi-strategy, a possible value being:
\begin{equation}
 \sum_{i=0}^\infty (1-p^{|S|})^i \cdot p^{|S|} \cdot i\cdot |S| \cdot \penmax\label{eqn:max-pen}
\end{equation}
where $p$ is the smallest non-zero probability assigned by $\delta$,
and $\penmax$ is the maximal local penalty over all states.
To see that \eqnref{eqn:max-pen} indeed gives a safe bound on $c$ (i.e. it is lower than any finite penalty achievable), observe that
for the penalty to be finite under a deterministic multi-strategy, for every
state $s$ there must be a path of length at most $|S|$ to a state from which no penalty
will be incurred. This path has probability at least $p^{|S|}$, and since the penalty accumulated
along a path of length $i\cdot |S|$ is at most $i\cdot |S| \cdot \penmax$, the properties of \eqnref{eqn:max-pen}
follow easily.

If the MILP problem has a solution, this is the
optimal dynamic penalty over all sound multi-strategies.
If not, no deterministic sound multi-strategy has a finite penalty
and the optimal penalty is $\infty$\label{page:inf-penalty}
(recall that we already established there is {\em some} sound multi-strategy).
In practice, we might choose a lower value of $c$ than the one above,
resulting in a multi-strategy that is sound, but possibly not optimally permissive.

\startpara{Correctness}
Formally, correctness of the MILP encoding for the case of dynamic penalties is captured by the following theorem.

\begin{thm}\label{thm:main-dynamic}
Let $\sgame$ be a game, $\phi=\property{r}{b}$ a property and $(\penbase,\dyn)$ a \emph{dynamic} penalty scheme.
Assume there is a sound multi-strategy for $\phi$.
The MILP formulation from~\figref{lp_encoding_dynamic} satisfies: (a) there is no solution if and only if the optimally permissive deterministic multi-strategy yields infinite penalty; and (b) there is a solution $\bar z_\sinit$ if and only if an optimally permissive deterministic multi-strategy yields
penalty $\bar z_\sinit$.
\end{thm}
\begin{proof}
We show that any sound multi-strategy with finite penalty $\bar z_\sinit$ gives rise to a satisfying assignment with the objective value $\bar z_\sinit$, and vice versa. Then, (b) follows directly, and (a) follows by the assumption that there is \emph{some} sound multi-strategy.

Let us prove that for any sound multi-strategy $\mstrat$ we can construct a satisfying assignment to
the constraints. For constraints \eqnref{eqn:det-fst} to \eqnref{eq:path-c}, the construction works exactly the same as in the proof of \thmref{thm:main-static}.
For the newly added variables, i.e. $z_s$ and $\ell_s$, we put $\bar\ell_s = \pendef{\penbase}{loc}{\mstrat,s}$, ensuring satisfaction of constraint~\eqnref{eq:dyn-diamond-a}, and:
\[
\bar z_s = \sup_{\sigma\Compl\mstrat,\pi} E_{\sgame,s}^{\ctrlstrat,\envstrat}(\penrew{\mstrat})
\]
which, together with \cite[Section 7.2.7, Equation 7.2.17]{Put94} (giving characterisation of optimal reward in terms of a linear program),
ensures that constraints~\eqnref{eq:dyn-diamond-b} and \ref{eq:dyn-box} are satisfied.

In the opposite direction, given a satisfying assignment we construct
$\mstrat$ for $\sgame'$ exactly as in the proof of \thmref{thm:main-static}. As before, we can argue that constraints \eqnref{eqn:det-fst} to \eqnref{eq:path-c}
are satisfied under any sound multi-strategy.
We now need to argue that the multi-strategy satisfies $\pendef{\penbase}{\dyn}{\mstrat,s} \ge \bar z_{\sinit}$.
It is easy to see that $\pendef{\penbase}{loc}{\mstrat,s} = \bar \ell_s$.
Moreover, by \cite[Section 7.2.7, Equation 7.2.17]{Put94} the penalty is the least solution to the
inequalities:
  \begin{align}
    z'_s &\ge \max_{a\in \mstrat(s)} \sum_{s' \in S} \delta(s,a)(s)\cdot z'_{s'} + \bar\ell_s&\text{for all $s\in S_\Diamond$}\label{eqn:dyn-proof-a}\\
    z'_s &\ge \max_{a\in \act(s)} \sum_{s' \in S} \delta(s,a)(s)\cdot z'_{s'}&\text{for all $s\in S_\Box$}\label{eqn:dyn-proof-b}
  \end{align}
We can replace \eqnref{eqn:dyn-proof-a} with:
\begin{align}
 z'_s \ge \max_{a\in \act(s)} \sum_{s' \in S} \delta(s,a)(s)\cdot z'_{s'} + \bar\ell_s - c\cdot(1-\bar y_{s,a})\label{eqn:dyn-proof-c}
\end{align}
since for $a\in \mstrat(s)$ we have $c\cdot(1 - \bar y_{s,a}) = 0$ and
otherwise $c\cdot(1-\bar y_{s,a})$ is greater than
$\sum_{s' \in S} \delta(s,a)(s)\cdot z'_{s'} + \bar\ell_s$
in the least solution to ~\eqnref{eqn:dyn-proof-a} and~\eqnref{eqn:dyn-proof-b}, by the definition of $c$.
Finally, it suffices to observe that the set of solutions to
\eqnref{eqn:dyn-proof-b} and~\eqnref{eqn:dyn-proof-c}
is the same as the set of solutions to \eqnref{eq:dyn-diamond-b} and~\eqnref{eq:dyn-box}.
\end{proof}




\subsection{Approximating Randomised Multi-Strategies}\label{sec:approximation}

In \sectref{sec:permissive_cont_synthesis}, we showed that
randomised multi-strategies can outperform deterministic ones.
The MILP encodings in \figfigref{lp_encoding_static}{lp_encoding_dynamic}, though,
cannot be adapted to the randomised case, since this would need non-linear constraints
(intuitively, we would need to multiply expected total rewards by probabilities of actions being allowed under a multi-strategy,
and both these quantities are unknowns in our formalisation).
%
Instead, in this section, we propose an \emph{approximation} which
finds the optimal randomised multi-strategy $\mstrat$ in which each
probability $\mstrat(s)(B)$ is a multiple of $\frac{1}{\gran}$ for a given
\emph{granularity} $\gran$. Any such multi-strategy can then be simulated
by a deterministic one on a transformed game, allowing synthesis to be carried out using
the MILP-based methods described in the previous section. Before giving
the definition of the transformed game, we show that we can simplify our
problem by restricting to multi-strategies
which in any state select at most two actions with non-zero probability.

\begin{thm}\label{2_set_rand_chain_condition}
Let $\sgame$ be a game, $\phi=\property{r}{b}$ a property, and $(\penbase,\pentype)$ a (static or dynamic) penalty scheme. For any sound
multi-strategy $\mstrat$ we can construct another sound multi-strategy
$\mstrat'$ such that $\pendef{\penbase}{\pentype}{\mstrat} \geq \pendef{\penbase}{\pentype}{\mstrat'}$ and
$|\support{\mstrat'(s)}| \leq 2$ for any $s \in S_{\Diamond}$.
\end{thm}
%

%

%



\begin{proof}
If the (dynamic) penalty under $\mstrat$ is infinite,
then the solution is straightforward: we can simply take $\mstrat'$ which, in every state,
allows a single action so that the reward is maximised. This restrictive multi-strategy enforces
a strategy that maximises the reward (so it performs at least as well as any other multi-strategy),
and at the same time it cannot yield the dynamic penalty worse than $\mstrat$, as the dynamic penalty under $\mstrat$
is already infinite.
From now on, we will assume that the penalty is finite.

Let $\mstrat$ be a multi-strategy allowing $n>2$ different sets
$A_1,\dots,A_n$ with non-zero probabilities $\lambda_1,\dots,\lambda_n$
in $s_1 \in S_{\Diamond}$. We construct a multi-strategy $\mstrat'$ that in
$s_1$ allows only two of the sets $A_1,\ldots,A_n$ with non-zero probability, and in other states behaves like $\mstrat$.

We first prove the case of dynamic penalties and then describe the differences for static penalties. Supposing that
$\inf_{\ctrlstrat \Compl\mstrat,\envstrat} E_{\sgame,s_1}^{\ctrlstrat,\envstrat}(r)
\le
\inf_{\ctrlstrat \Compl\mstrat',\envstrat} E_{\sgame,s_1}^{\ctrlstrat,\envstrat}(r)
$,
we have that the total reward is:
\begin{align*}
 \inf_{\ctrlstrat \Compl\mstrat,\envstrat}
  E_{\sgame,\sinit}^{\ctrlstrat,\envstrat}(r)
  & = \inf_{\ctrlstrat\Compl\mstrat,\envstrat} \Big( E_{\sgame,\sinit}^{\ctrlstrat,\envstrat}(r{\downarrow}s_1) + \Prb_{\sgame,\sinit}^{\ctrlstrat,\envstrat}(\future s_1)
   \cdot E_{\sgame,s_1}^{\ctrlstrat,\envstrat}(r)\Big)\\
  & = \inf_{\ctrlstrat\Compl\mstrat,\envstrat} \Big( E_{\sgame,\sinit}^{\ctrlstrat,\envstrat}(r{\downarrow}s_1)
   + \Prb_{\sgame,\sinit}^{\ctrlstrat,\envstrat}(\future s_1)
   \cdot \inf_{\ctrlstrat'\Compl\mstrat,\envstrat'} E_{\sgame,s_1}^{\ctrlstrat',\envstrat'}(r)\Big)\\
  & \le \inf_{\ctrlstrat\Compl\mstrat,\envstrat} \Big( E_{\sgame,\sinit}^{\ctrlstrat,\envstrat}(r{\downarrow}s_1)
   + \Prb_{\sgame,\sinit}^{\ctrlstrat,\envstrat}(\future s_1)
   \cdot \inf_{\ctrlstrat'\Compl\mstrat',\envstrat'} E_{\sgame,s_1}^{\ctrlstrat',\envstrat'}(r)\Big)\\
  & = \inf_{\ctrlstrat\Compl\mstrat',\envstrat} \Big( E_{\sgame,\sinit}^{\ctrlstrat,\envstrat}(r{\downarrow}s_1)
   + \Prb_{\sgame,\sinit}^{\ctrlstrat,\envstrat}(\future s_1)
   \cdot \inf_{\ctrlstrat'\Compl\mstrat',\envstrat'} E_{\sgame,s_1}^{\ctrlstrat',\envstrat'}(r)\Big)\tag{$\ast$}\\
   &= \inf_{\ctrlstrat\Compl\mstrat',\envstrat}E_{\sgame,\sinit}^{\ctrlstrat,\envstrat}(r)
\end{align*}
where the equation ($\ast$) above follows by the fact that, up to the first time $s_1$ is reached, $\mstrat$ and $\mstrat'$
allow the same actions.
Hence, it suffices to define $\mstrat'$ so that
$\inf_{\ctrlstrat\Compl\mstrat,\envstrat} E_{\sgame,s_1}^{\ctrlstrat,\envstrat}(r)
\le
\inf_{\ctrlstrat\Compl\mstrat',\envstrat} E_{\sgame,s_1}^{\ctrlstrat,\envstrat}(r)
$. Similarly, for the penalties, it is enough to ensure
$\sup_{\ctrlstrat\Compl\mstrat,\envstrat} E_{\sgame,s_1}^{\ctrlstrat,\envstrat}(\penrew{\mstrat})
\ge
\sup_{\ctrlstrat\Compl\mstrat',\envstrat} E_{\sgame,s_1}^{\ctrlstrat,\envstrat}(\penrew{\mstrat'})
$.

Let $P_i$ and $R_i$, where $i \in \{1, ..., n\}$, be the penalties and
rewards from $\mstrat$ after allowing $A_i$ against an optimal
opponent strategy, i.e.:
\begin{align*}
 P_i & = \sum_{a\notin A_i}\penbase(s_1,a) + \sup_{\ctrlstrat\Compl\mstrat,\envstrat}\max_{a\in A_i}\sum_{s'\in S}\delta(s_1,a)(s')\cdot E_{\sgame,s'}^{\ctrlstrat,\envstrat}(\penrew{\mstrat})\\
 R_i & = \inf_{\ctrlstrat\Compl\mstrat,\envstrat}\min_{a\in A_i}(r(s_1,a) + \sum_{s'\in S}\delta(s_1,a)(s')\cdot E_{\sgame,s'}^{\ctrlstrat,\envstrat}(r))
\end{align*}
We also define $R=\inf_{\ctrlstrat\Compl\mstrat,\envstrat} E_{\sgame,s_1}^{\ctrlstrat,\envstrat}(r)$
and $P=\sup_{\ctrlstrat\Compl\mstrat,\envstrat} E_{\sgame,s_1}^{\ctrlstrat,\envstrat}(\penrew{\mstrat})$
and have $R = \sum_{i=1}^n \lambda_i R_i$ and $P = \sum_{i=1}^n \lambda_i P_i$.


%
Let $S_0\subseteq S$ be those states for which there are
$\ctrlstrat\Compl\mstrat$ and $\pi$
ensuring a return to $s_1$ without accumulating any reward.
Formally, $S_0$ contains all states $s_0$ which satisfy $\Prb_{\sgame,s_0}^{\ctrlstrat,\envstrat}(\future \{s_1\})=1$ and $E_{\sgame,s_0}^{\ctrlstrat,\envstrat}(r{\downarrow}s_1)$ for some $\ctrlstrat\Compl\mstrat$ and $\pi$.
%
We say that $A_i$ is {\em progressing} if for
all $a\in A_i$ we have $r(s_1,a)>0$ or $\support{\delta(s_1,a)}\not\subseteq S_0$.
We note that $A_i$ is progressing whenever
$R_i>R$ (since any $a$ violating the condition above could have been used by the
opponent to force $R_i\le R$).

For each tuple $\mu = (\mu_1,\dots,\mu_n)\in\mathbb{R}^n$, let
$R^\mu = \mu_1R_1+\dots+\mu_nR_n$ and $P^\mu =
\mu_1P_1+\dots+\mu_nP_n$. Then the set $T = \{(R^\mu, P^\mu) \mid 0\le\mu_i\le
1, \mu_1+\dots+\mu_n=1\}$ is a bounded convex polygon, with vertices
given by images $(R^{e_i},P^{e_i})$ of unit vectors (i.e., Dirac distributions)
$e_i=(0,\dots,0,1,0,\dots,0)$, and containing
$(R^\lambda,P^\lambda) = (R,P)$. To each vertex
$(R_j,P_j)$ we associate the (non-empty) set $I_j = \{i\mid (R^{e_i},P^{e_i}) = (R_j,P_j)\}$ of
indices.

We will find $\alpha \in (0,1)$ and $1\le u,v \le n$ such that one of $A_u$ or $A_v$ is progressing, and define the multi-strategy $\mstrat'$ to pick $A_u$ and $A_v$ with
probabilities $\alpha$ and $1-\alpha$, respectively. We distinguish several cases, depending on the shape of $T$:
\begin{enumerate}
  \item $T$ has non-empty interior. Let $(R_1,P_1),\dots,(R_m,P_m)$ be
    its vertices in the anticlockwise order. Since all $\lambda_i$ are positive,
    $(R,P)$ is in the interior of $T$.
    Now consider the point $(R,P')$ directly below
    $(R,P)$ on the boundary of $T$, i.e. $P' = \min\{P''\mid (R,P'')\in T\}$.
    If $(R,P')$ is not a vertex, it is a
    convex combination of adjacent vertices $(R,P')=\alpha(R_j,P_j) + (1-\alpha)(R_{j+1},P_{j+1})$,
    and we pick such $\alpha$ and $u\in I_j$ and $v\in I_{j+1}$.
    If $(R,P')$ happens to be a vertex $(R_j,P_j)$ we can
    (since $P_j<P$) instead choose sufficiently small $\alpha>0$ so that $R \ge \alpha R_j + (1-\alpha) R_{j+1}$
    and $P \le \alpha P_j + (1-\alpha) P_{j+1}$ and again pick $u\in I_j$ and $v\in I_{j+1}$.
    In either case, we necessarily have $R_{j+1} > R$ (by ordering of the vertices in the anticlockwise order and since $(R,P)$ is in the interior of $T$),
    and so $A_v$ is progressing.

  \item $T$ is a vertical line segment, i.e. it is the convex hull of
    two extreme points $(R,P_0)$ and $(R,P_1)$ with $P_0<P_1$.  In
    case $R=0$, we can simply always allow some $A_i$ with $i\in I_0$,
    minimising the penalty and still achieving reward $0$.

    If $R>0$, there must be at least one progressing $A_u$.
    Since all $\lambda_i$ are positive, $(R,P)$ lies inside the
    line segment, and in particular $P>P_0$.  We can therefore choose
    some $v$ and $\alpha\in (0,1)$ such that $P \le \alpha \cdot P_u + (1-\alpha) \cdot P_v$.
  \item $T$ is a non-vertical line segment, i.e. it is the convex hull
    of two extreme points $(R_0,P_0)$ and $(R_1,P_1)$ with $R_0<R_1$.
    Since all $\lambda_i$ are positive, $(R,P)$ is not one of the
    extreme points, i.e. $(R,P) = \alpha(R_0,P_0) +
    (1-\alpha)(R_1,P_1)$ with $0<\alpha<1$. We can therefore choose
    $u\in I_0,v\in I_1$.  Again,
    since $R_1>R$, $A_v$ is progressing.
  \item $T$ consists of a single point $(R,P)$. This can be treated like the
    second case: either $R=0$, and we can allow any combination, or $R>0$, and
    there is some progressing $A_u$, and we then pick arbitrary $v$ and $\alpha$.
\end{enumerate}

\noindent We now want to show that the reward of the updated
multi-strategy is indeed no worse than before.
For $i\in \{u,v\}$ we define:
\[
 R'_i = \inf_{\ctrlstrat\Compl\mstrat',\envstrat}\min_{a\in A_i}(r(s_1,a) + \sum_{s'\in S}\delta(s_1,a)(s')\cdot E_{\sgame,s'}^{\ctrlstrat,\envstrat}(r))
\]
and we define $R' = \alpha R'_u + (1-\alpha) R'_v$.
Pick an action $a$ (resp. $a'$) that realises the minimum and strategies $\sigma$ and $\pi$ (resp. $\sigma'$ and $\pi'$)
that realise the infimum in the definition of $R_i$ (resp. $R'_i$).
(Such strategies indeed exist).
Define:
\begin{align*}
 c_i &= \sum_{s'} \delta(s_1,a)(s') \cdot \Prb_{\sgame,s'}^{\ctrlstrat,\envstrat}(\future s_1)&
 c'_i &= \sum_{s'} \delta(s_1,a')(s') \cdot \Prb_{\sgame,s'}^{\ctrlstrat,\envstrat}(\future s_1)\\
 d_i &= r(s,a) + \sum_{s'} \delta(s_1,a)(s') \cdot E_{\sgame,s'}^{\ctrlstrat,\envstrat}(r{\downarrow}s_1)&
 d'_i &= r(s,a') + \sum_{s'} \delta(s_1,a')(s') \cdot E_{\sgame,s'}^{\ctrlstrat,\envstrat}(r{\downarrow}s_1)
\end{align*}
We have $R_i = c_i \cdot R + d_i$ for every $1\le i \le n$, and $R'_i = c'_i \cdot R' + d'_i$ for all $i\in \{u,v\}$.
Then:
\begin{equation*}
  \begin{split}
    R' - R &= (\alpha R'_u + (1-\alpha)R'_v) - \sum\nolimits_i\lambda_iR_i\\
    &= (\alpha R'_u + (1-\alpha)R'_v) - (\alpha R_u + (1-\alpha)R_v) + (\alpha R_u + (1-\alpha)R_v) - \sum\nolimits_i\lambda_iR_i\\
        &\ge (\alpha R'_u + (1-\alpha)R'_v) - (\alpha R_u + (1-\alpha)R_v)\\
        &\  \hspace*{15pt}\mbox{ (by the choice of $\alpha,u,v$)}\\
    &= (\alpha(c'_uR' + d'_u) + (1-\alpha)(c'_vR' + d'_v)) - (\alpha(c_uR + d_u) + (1-\alpha)(c_vR + d_v))\\
    &\ge (\alpha(c'_uR' + d'_u) + (1-\alpha)(c'_vR' + d'_v)) - (\alpha(c'_uR + d'_u) + (1-\alpha)(c'_vR + d'_v))\\
    &\  \hspace*{15pt}\mbox{ ($c_iR+d_i\leq c'_iR + d'_i$ by optimality of actions/strategies defining $c_i$ and $d_i$)}\\
    &= (\alpha c'_u + (1-\alpha)c'_v)(R' - R),
  \end{split}
\end{equation*}
i.e., $(1 - \alpha c'_u - (1-\alpha)c'_v) (R' - R) \ge 0$.
By finiteness of rewards and the choice of $\mstrat(s_1)$, at least one of the return probabilities $c'_u,c'_v$ is less than $1$, and thus so is $\alpha c'_u + (1-\alpha)c'_v$, therefore $R'\ge R$.

We can show that the penalty under $\mstrat'$ is at most as big as the penalty under $\mstrat$ in exactly the same way (note that in addition using $\penrew{\mstrat'}$ instead of $\penrew{\mstrat}$ for  $c'$, $d'$, $R'$ and $R'_i$).
For static penalties, the proof that the new multi-strategy is no worse
than the old one is straightforward from the choice of
$\mstrat'(s_1)$.
\end{proof}

The result just proved allows us to simplify the game construction that we use
to map between (discretised) randomised multi-strategies and deterministic ones.
Let the original game be $\sgame$ and the transformed game be $\sgame'$.
The transformation is illustrated in \figref{fig:approximation}.
The left-hand side shows a controller state $s \in S_{\Diamond}$ in the original game $\sgame$
(i.e., the one for which we are seeking randomised multi-strategies).
For each such state, we add the two layers of states illustrated on the
right-hand side of the figure:
\emph{gadgets} $s'_1,s'_2$ 
representing the two subsets $B\subseteq\act(s)$ with
$\mstrat(s)(B)>0$, and \emph{selectors} $s_i$ (for $1\le i \le m$),
which distribute probability among the two gadgets.
Two new actions, $b_1$ and $b_2$, are also added
to label the transitions between selectors $s_i$ and gadgets $s'_1,s'_2$.

The selectors $s_i$ are reached from $s$
via a transition using fixed probabilities $p_1,\dots,p_m$ which need
to be chosen appropriately.
For efficiency, we want to minimise the number of selectors $m$ for each state $s$.
We let $m=\lfloor 1+\log_2 \gran \rfloor$ and $p_i=\frac{l_i}{\gran}$,
where $l_1\,\dots,l_m\in\Nset$ are defined recursively as follows:
$l_1 = \lceil\frac{M}{2}\rceil$
and $\smash{l_i = \lceil\frac{\gran-(l_1+\dots+l_{i-1})}{2}\rceil}$ for $2\leq i\leq m$.
For example, for $M{=}10$, we have $m=4$ and $l_1,\dots,l_4=5,3,1,1$, so
$p_1,\dots,p_4=\frac{5}{10},\frac{3}{10},\frac{1}{10},\frac{1}{10}$.

\input{fig_transformation}

We are now able to find optimal discretised randomised multi-strategies in $\sgame$
by finding optimal deterministic multi-strategies in $\sgame'$.
This connection will be formalised in \lemref{transformed_mstrats} below.
But we first point out that, for the case of \emph{static} penalties,
a small transformation to the MILP encoding (see \figref{lp_encoding_static})
used to solve game $\sgame'$ is required.
The problem is that the definition of static penalties on $\sgame'$
does not precisely capture the static penalties of the original game $\sgame$.
In this, case we adapt \figref{lp_encoding_static} as follows.
For each state $s$, action $a\in A(s)$ and $i\in\{1,\dots,n\}$,
we add a binary variable $y'_{s_i,a}$ and
constraints $y'_{s_i,a} \ge y_{s'_j,a} - (1-y_{s_i,b_j})$ for $j\in\{1,2\}$.
We then change the objective function that we minimise to:
\[
- x_{\sinit}\  + \sum\nolimits_{s \in S_{\Diamond}} \sum\nolimits_{a \in A(s)}\sum\nolimits_{i=1}^m p_i\cdot(1-y'_{s_i,a})\cdot\penbase(s,a)
\]
\begin{lem}\label{transformed_mstrats}
 Let $\sgame$ be a game, $\phi=\property{r}{b}$ a property, $(\penbase,\pentype)$ a (static or dynamic) penalty scheme,
 and let $\sgame'$ be the game transformed as described above. The following two claims are equivalent:
 \begin{enumerate}
  \item There is a sound multi-strategy $\mstrat$ in $\sgame$ with $\pendef{\penbase}{\dyn}{\mstrat} = x$ (or, for static penalties, $\pendef{\penbase}{\sta}{\mstrat} = x$), and $\mstrat$ only uses probabilities that are multiples of $\frac{1}{M}$.
  \item There is a sound deterministic multi-strategy $\mstrat'$ in $\sgame'$ and $\pendef{\penbase}{\dyn}{\mstrat} = x$
  (or, for static penalties,
  $\sum\nolimits_{s \in S_{\Diamond}} \sum_{i=1}^m p_i \cdot \psi'(i) = x$ where
  $\psi'(i)$ is $\sum_{a\in A(s) \setminus (\mstrat'(s'_1) \cup \mstrat'(s'_2))} \penbase(s,a)$ if $\mstrat'(s_i)=\{b_1,b_2\}$, and 
   otherwise $\psi'(i)$ is $\sum_{a\in A(s) \setminus \mstrat'(s'_j)} \penbase(s,a)$ where $j$ is the (unique) number with $b_j \in\mstrat'(s_i)$).
 \end{enumerate}
\end{lem}
\begin{proof}
 Firstly, observe that for any integer $0\le k \le M$ there is a set $I_{k} \subseteq \{1,\ldots,m\}$
 with $\sum_{j \in I_{k}} l_{j} = k$. The opposite direction also holds.

 Let $\mstrat$ be a multi-strategy in $\sgame$.
 By \thmref{2_set_rand_chain_condition} we can assume that $|\support{\mstrat(s)}| \leq 2$ for any $s$.
 We create $\mstrat'$ as follows. For every state $s\in S_\Diamond$ with $\{A_1,A_2\} = \support{\mstrat(s)}$,
 we set $\mstrat'(s'_1)(A_1) = 1$ and
 $\mstrat'(s'_2)(A_2) = 1$. Then, supposing $\mstrat(s)(A) = \frac{k}{M}$, we
 let $\mstrat'(s_i)(\{b_1\})=1$ whenever $i\in I_k$,
 and $\mstrat'(s_i)(\{b_2\})=1$ whenever $i\not\in I_k$. If $\mstrat(s)$ is a singleton set, the construction is analogous.
 Clearly, the property for static penalties is preserved.
 For any memoryless $\sigma'\Compl\mstrat'$ there is
 a memoryless strategy $\sigma\Compl\mstrat$ that is given by $\sigma(s)(a) = \frac{k}{M} \cdot \sigma(s'_1)(a) + (1-\frac{k}{M}) \cdot \sigma(s'_1)(a)$ for any $a$,
 and conversely for any $\sigma\Compl\mstrat$ we can define $\sigma'\Compl\mstrat'$
 by putting $\sigma'(s'_1)=d_{s,A_1}$ and $\sigma'(s'_2)=d_{s,A_2}$ for all $s$, where $d_{s,A_1}$
 and $d_{s,A_2}$ are distributions witnessing that $\sigma$ is compliant with $\mstrat$.
 It is easy to see that both $\sigma$ and $\sigma'$ in either of the above yield the same reward and dynamic penalty.

 In the other direction, we define $\mstrat$ from $\mstrat'$ for all $s\in S_\Diamond$ as follows.
 Let $A_1$ and $A_2$ be the sets allowed by $\mstrat'$ in $s'_1$ and $s'_2$ respectively. If $A_1 = A_2$, then $\mstrat(s)$ allows this set
 with probability $1$. Otherwise $\mstrat(s)$ allows
 the set $A_1\cup A_2$ with probability $\sum_{i : \mstrat(s_i) = \{b_1,b_2\}} p_i$,
 the set $A_1$ with probability $\sum_{i : \mstrat(s_i) = \{b_1\}} p_i$ and
 the set $A_2$ with probability $\sum_{i : \mstrat(s_i) = \{b_2\}} p_i$.
 The correctness can be proved similarly to above.
\end{proof}



Our next goal is to show that, by varying the granularity $\gran$, we
can get arbitrarily close to the optimal penalty for a randomised
multi-strategy and, for the case of static penalties, define a suitable choice of $\gran$.
This will be formalised in \thmref{thm:approx} shortly.
First, we need to establish the following intermediate result,
stating that, in the static case, in addition to
\thmref{2_set_rand_chain_condition} we can require the
action subsets allowed by a multi-strategy to be ordered with respect to the subset relation.

\begin{thm}\label{2_set_rand_subset}
Let $\sgame$ be a game, $\phi=\property{r}{b}$ a property and $(\penbase,\sta)$ a \emph{static} penalty scheme.
For any sound
multi-strategy $\mstrat$ we can construct another sound multi-strategy
$\mstrat'$ such that $\pendef{\penbase}{\sta}{\mstrat} \geq \pendef{\penbase}{\sta}{\mstrat'}$ and
for each $s \in S_{\Diamond}$,
if $\support{\mstrat'(s)}{=}\{B,C\}$, then either $B\subseteq C$ or $C\subseteq B$.
\end{thm}
\begin{proof}
%
Let $\mstrat$ be a multi-strategy and fix $s_1$ such that $\mstrat$ takes two different actions $B$ and $C$ with probability $p\in(0,1)$ and $1-p$ where $B\nsubseteq C$ and $C\nsubseteq B$.
If $\inf_{\sigma\Compl\mstrat,\pi} E_{\sgame,s_1}^{\sigma,\pi}(r) = 0$,
then we can in fact allow deterministically the single set $A(s_1)$ and we are done. Hence, suppose that the reward accumulated from $s_1$ is non-zero.

Suppose, w.l.o.g., that:
\begin{equation}\label{eqn:subset-comparison}
 \min_{a\in B}r(s_1,a) + \sum_{s'\in S} \delta(s_1,a)(s') \cdot  \inf_{\sigma\Compl\mstrat,\pi} E_{\sgame,s'}^{\sigma,\pi}(r) \le
 \min_{a\in C}r(s_1,a) + \sum_{s'\in S} \delta(s_1,a)(s') \cdot  \inf_{\sigma\Compl\mstrat,\pi} E_{\sgame,s'}^{\sigma,\pi}(r)
\end{equation}
It must be the case that, for some $D\in \{B,C\}$, we have:
\begin{equation}\label{eqn:subset-nonzero}
 \min_{a\in D} r(s_1,a) + \inf_{\sigma\Compl\mstrat,\pi} \sum_{s'\in S} \delta(s_1,a)(s')\cdot E_{\sgame,s'}^{\sigma,\pi}(r{\downarrow}s_1) > 0
\end{equation}
(otherwise the minimal reward accumulated from $s_1$ is $0$ since there is a compliant strategy that keeps returning to $s_1$ without ever accumulating any reward), and if the inequality in \eqnref{eqn:subset-comparison} is strict, then \eqnref{eqn:subset-nonzero} holds for $D=C$. W.l.o.g., suppose that
the above property holds for $C$. We define $\mstrat'$ by modifying $\mstrat$ and picking $B\cup C$ with probability
$p$, $C$ with $(1-p)$, and $B$ with probability $0$.

Under $\mstrat$, the minimal reward achievable by some compliant strategy is given as the least solution to the following equations \cite[Theorem 7.3.3]{Put94} (as before, the notation of \cite{Put94} requires ``negative'' models):
\begin{align*}
 x_s &= \sum_{A\in \support{\mstrat(s)}} \mstrat(s)(A) \cdot \min_{a\in A} \sum_{s'\in S} r(s,a) + \delta(s,a)(s') \cdot x_{s'} &\text{for $s\in S_\Diamond$}\\
 x_s &= \min_{a\in A(s)} \sum_{s'\in S} r(s,a) + \delta(s,a)(s') \cdot x_{s'} &\text{for $s\in S_\Box$}
\end{align*}
The minimal rewards $x'_s$ achievable under $\mstrat'$ are defined analogously.
In particular, for the equation with $s_1$ on the left-hand side we have:
\begin{align*}
 x_{s_1} &= p\cdot \min_{a\in B} r(s_1,a) + \sum_{s'\in S} \delta(s_1,a)(s') \cdot x_{s'} + (1-p) \cdot \min_{a\in C} r(s_1,a) + \sum_{s' \in S} \delta(s_1,a)(s') \cdot x_{s'}\\ 
 x'_{s_1} &= p\cdot \min_{a\in B \cup C} r(s_1,a) + \sum_{s'\in S} \delta(s_1,a)(s') \cdot x'_{s'} + (1-p) \cdot \min_{a\in C} r(s_1,a) + \sum_{s'\in S} \delta(s_1,a)(s') \cdot x'_{s'} 
\end{align*}
We show that the least solution $\bar x$ to $x$ is also the least solution to $x'$.

First, note that $\bar x$ is clearly a solution to any equation with $s\neq s_1$
on the left-hand side since these equations remain unchanged in both sets of equations.
As for the equation with $s_1$, we have $\min_{a\in B} \sum_{s'} r(s_1,a) + \delta(s_1,a)(s') \cdot \bar x_{s'} \le \min_{a\in C} \sum_{s'} r(s_1,a) + \delta(s_1,a)(s') \cdot \bar x_{s'}$,
and so necessarily $\min_{a\in B} \sum_{s'} r(s_1,a) +\delta(s_1,a)(s') \cdot \bar x_{s'} = \min_{a\in B \cup C} \sum_{s'} r(s_1,a) + \delta(s_1,a)(s') \cdot \bar x_{s'}$.

To see that $\bar x$ is the {\em least} solution to $x'$, we show that (i) for all $s$, if
$\inf_{\sigma\Compl\mstrat', \pi} E_{\sgame,s}^{\sigma,\pi}(r) = 0$ then $\bar x_s = 0$; and
(ii) there
is a unique fixpoint satisfying $\bar x_s = 0$ for all $s$ such that
$\inf_{\sigma\Compl\mstrat', \pi} E_{\sgame,s}^{\sigma,\pi}(r) = 0$.

For (i), suppose $\bar x_s > 0$. Let $\sigma'$ be a strategy compliant with $\mstrat'$, and $\pi$ an arbitrary strategy.
Suppose $\Prb_{\sgame,s}^{\sigma',\pi}(\future s_1) = 0$, then there is a strategy $\sigma$
compliant with $\mstrat$ which behaves exactly as $\sigma'$ when starting from $s$, and by
our assumption on the properties of $\bar x_s$ we get that $E_{\sgame,s}^{\sigma,\pi}(r) > 0$ and so $E_{\sgame,s}^{\sigma',\pi}(r) > 0$.
Now suppose that $\Prb_{\sgame,s}^{\sigma',\pi}(\future s_1) > 0$. For this case, by condition~\eqnref{eqn:subset-nonzero},
the fact that it holds for $D=C$ and by defining $\mstrat'$ so that it picks $C$ with nonzero probability we
get that the reward under any strategy compliant with $\mstrat'$ is
non-zero when starting in $s_1$, and so $E_{\sgame,s}^{\sigma',\pi}(r) > 0$.

Point (ii) can be obtained by an application of~\cite[Proposition 7.3.4]{Put94}.
%
\end{proof}

We can now return to the issue of how to vary the granularity $\gran$
to get sufficiently close to the optimal penalty.
We formalise this as follows.

\begin{thm}\label{thm:approx}
Let $\sgame$ be a game, $\phi=\property{r}{b}$ a property, and $(\penbase,\pentype)$ a (static or dynamic) penalty scheme.
Let $\mstrat$ be a sound multi-strategy. For any $\varepsilon>0$,
there is an $\gran$ and a sound multi-strategy $\mstrat'$ of granularity $\gran$ satisfying
$\pendef{\penbase}{\pentype}{\mstrat'} - \pendef{\penbase}{\pentype}{\mstrat} \le \varepsilon$.
Moreover, for static penalties it suffices to take
$\gran = \lceil \sum_{s\in S,a\in\act(s)}\frac{\penbase(s,a)}{\varepsilon}\rceil$.
\end{thm}
%
%
%
%
\begin{proof}

By \thmref{2_set_rand_chain_condition}, we can assume
$\support{\mstrat(t)} = \{A_{1},A_{2}\}$ for any state $t\in S_\Diamond$.

We deal with the cases of static and dynamic penalties separately.
For static penalties, let $t\in S_\Diamond$ and $\mstrat(t)(A_{1}) = q$, $\mstrat(t)(A_{2}) = 1-q$
  for $A_{1}\subseteq A_{2}\subseteq\act(t)$. Modify $\mstrat$
  by rounding $q$ up to the number $q'$ which is the
  nearest multiple of $\frac{1}{\gran}$. The resulting multi-strategy $\mstrat'$ is again sound,
  since any strategy compliant with $\mstrat'$ is also compliant with $\mstrat$: the witnessing
  distributions (see Definition~\ref{def:compliant}) $d_{t,A_1}$ and $d_{t,A_2}$ for $\mstrat$ are obtained
  from the distributions $d'_{t,A_1}$ and $d'_{t,A_2}$ for $\mstrat'$ by setting
  $d_{t,A_1}(a) = d'_{t,A_1}(a)$ for all $a\in A_1$
  and
  $d_{t,A_2}(a) = \frac{q'-q}{1-q}\cdot d'_{t,A_1}(a) + \frac{1-q'}{1-q}\cdot d'_{t,A_2}(a)$ for all $a\in A_2$;
  note that both $d_{t,A_1}$ and $d_{t,A_2}$ are indeed probability distributions.
  Further, the penalty in $\mstrat'$
  changes by at most
  $\frac{1}{\gran}\sum_{a\in\act(s)}\penbase(t,a)$. To obtain the result we repeat the above for all $t$.

Now let us consider dynamic penalties. Intuitively, the claim follows since by making small changes to the multi-strategy,
 while not (dis)allowing any new actions, we only cause small changes to the reward and penalty.

Let $\mstrat$ be a multi-strategy and $t\in S_\Diamond$ a state. W.l.o.g., suppose:
 \[
  \inf_{\ctrlstrat \Compl \mstrat, \pi}\min_{a\in A_{1}} r(s,a) + \sum_{s'} \delta(s,a)(s') \cdot E_{\sgame,s'}^{\ctrlstrat,\envstrat}(r) \ge \inf_{\ctrlstrat \Compl \mstrat, \pi}\min_{a\in A_{2}} r(s,a) + \sum_{s'} \delta(s,a)(s') \cdot E_{\sgame,s'}^{\ctrlstrat,\envstrat}(r)
\]
 For $0 < x < \mstrat(t)(A_2)$,
 we define a multi-strategy $\mstrat_x$
 by $\mstrat_x(t)(A_{1}) = \mstrat(t)(A_{1}) + x$ and
 $\mstrat_x(t)(A_{2}) = \mstrat(t)(A_{2}) - x$ 
 , and  $\mstrat_x(s) = \mstrat(s)$ for all $s\neq t$. We will show that
 $\inf_{\ctrlstrat \Compl \mstrat_x, \pi} E_{\sgame,\sinit}^{\ctrlstrat,\envstrat}(r)
  \ge \inf_{\ctrlstrat \Compl \mstrat, \pi} E_{\sgame,\sinit}^{\ctrlstrat,\envstrat}(r)$.
 Consider the following functional $F_x: (S \rightarrow \Rset) \rightarrow (S\rightarrow \Rset)$, constructed for the multi-strategy $\mstrat_x$
 \begin{align*}
  F_x(f)(s) &= \sum_{A\in \support{\mstrat(s)}} \mstrat_x(s) \cdot \min_{a\in A} \sum_{s'\in S} r(s,a) + \delta(s,a)(s') \cdot f(s') &\text{for $s\in S_\Diamond$}\\
  F_x(f)(s) &= \min_{a\in A(s)} \sum_{s'\in S} r(s,a) + \delta(s,a)(s') \cdot f(s') &\text{for $s\in S_\Box$}
 \end{align*}
 Let $f$ be the function assigning $\inf_{\ctrlstrat \Compl \mstrat, \pi} E_{\sgame,t}^{\ctrlstrat,\envstrat}(r)$ to $s$.
 Observe that $f(s)=0$ whenever $\inf_{\ctrlstrat \Compl \mstrat_x, \pi} E_{\sgame,t}^{\ctrlstrat,\envstrat}(r)$; this follows since
 $x < \min\{\mstrat(t)(A_1),\mstrat(t)(A_2)\}$ and so both $\mstrat$ and $\mstrat_x$ allow the same actions with non-zero probability. Also,
 $F_x(f)(s) \ge f(s)$: for $s\neq t$ in fact $F_x(f)(s) = f(s)$ because the corresponding functional $F$ for $\mstrat$ coincides with $F_x$ on $s$;
 for $s = t$, we have $F_x(f)(s) \ge f(s)$ since $\min_{a\in A_1} r(s,a) + \sum_{s'} \delta(s,a)(s') \cdot f(s')\ge \min_{a\in A_2} r(s,a) + \sum_{s'} \delta(s,a)(s') \cdot f(s')$ by the properties of $A_1$ and $A_2$ and since $x$ is non-negative. Hence, we can apply \cite[Proposition 7.3.4]{Put94} and obtain that $\mstrat_x$ ensures at least the same reward as $\mstrat$.
 Thus, by increasing the probability of allowing $A_{1}$ in $t$ the soundness of the multi-strategy is preserved.


 Further, for any strategy $\sigma'$ compliant with $\mstrat_x$ and any $\envstrat$, the penalty when starting in $t$, i.e. $E_{\sgame,t}^{\sigma',\envstrat}(\penrew{\mstrat_x})$, is equal to:
 \[
   \pendef{\penbase}{loc}{\mstrat_x,t} + \sum_{a\in A} \xi'(t)(a) \sum_{t'\in S} \delta(t,a)(t') \cdot \big(E_{\sgame,t'}^{\sigma',\envstrat}(\penrew{\mstrat_x} {\downarrow} t) + \Prb_{\sgame,t'}^{\sigma',\envstrat}(\future t) \cdot E_{\sgame,t}^{\sigma',\envstrat}(\penrew{\mstrat_x})\big)
 \]
for $\xi' = \sigma'\cup \pi$. 
 There is a strategy $\sigma$ compliant with $\mstrat$ which
 differs from $\sigma'$ only on $t$, where $\sum\nolimits_{a\in A}|\sigma'(s,a)-\sigma(s,a)| \le x$. We have, for any $\envstrat$:
 \begin{align*}
  \lefteqn{E_{\sgame,t}^{\ctrlstrat,\envstrat}(\penrew{\mstrat})}\\ &= \pendef{\penbase}{loc}{\mstrat,t} + \sum_{a\in A} \xi(t,a) \sum_{s\in S} \delta(t,a)(s) \cdot (E_{\sgame, s}^{\ctrlstrat,\envstrat}(\penrew{\mstrat} {\downarrow} t) + \Prb_{\sgame,s}^{\ctrlstrat,\envstrat}(\future t) \cdot E_{\sgame,t}^{\ctrlstrat,\envstrat}(\penrew{\mstrat}))\\
   &\ge \pendef{\penbase}{loc}{\mstrat_x,t} {-} x {+} \sum_{a\in A} (\xi'(t,a) {-} x) \sum_{s\in S} \delta(t,a)(s) {\cdot} (E_{\sgame,s}^{\ctrlstrat',\envstrat}(\penrew{\mstrat_x} {\downarrow} t) {+} \Prb_{\sgame,s}^{\ctrlstrat',\envstrat}(\future t) {\cdot} E_{\sgame,t}^{\ctrlstrat,\envstrat}(\penrew{\mstrat_x})) 
 \end{align*}
 where $\xi=\sigma\cup \envstrat$ and the rest is as above.

 Thus:
 \begin{eqnarray*}
  E_{\sgame,t}^{\ctrlstrat',\envstrat}(\penrew{\mstrat_x}) &=&  \frac{\pendef{\penbase}{loc}{\mstrat_x,t} + \sum_{a\in A} \xi'(t)(a) \sum_{t'\in S} \delta(t,a)(t') \cdot E_{\sgame,t'}^{\sigma',\envstrat}(\penrew{\mstrat_x}{\downarrow}t)}{1-\sum_{a\in A} \xi'(t)(a) \sum_{t'\in S} \delta(t,a)(t') \cdot \Prb_{\sgame,t'}^{\sigma',\envstrat}(\future t)}\\
  E_{\sgame,t}^{\ctrlstrat,\envstrat}(\penrew{\mstrat}) &\ge&  \frac{\pendef{\penbase}{loc}{\mstrat,t} - x + \sum_{a\in A} (\xi'(t)(a) - x) \sum_{t'\in S} \delta(t,a)(t') \cdot E_{\sgame,t'}^{\sigma',\envstrat}(\penrew{\mstrat}{\downarrow}t)}{1-\sum_{a\in A} (\xi'(t)(a) - x) \sum_{t'\in S} \delta(t,a)(t') \cdot \Prb_{\sgame,t'}^{\sigma',\envstrat}(\future t)}\\
 \end{eqnarray*}
and so $E_{\sgame,t}^{\ctrlstrat',\envstrat}(\penrew{\mstrat_x}) - E_{\sgame,t}^{\ctrlstrat,\envstrat}(\penrew{\mstrat})$ goes to $0$ as $x$ goes to $0$. Hence,
$\pendef{\penbase}{\dyn}{\mstrat_x} - \pendef{\penbase}{\dyn}{\mstrat}$ goes to $0$ as $x$ goes to $0$.
 
The above gives us that, for any error bound $\varepsilon$ and a fixed state $s$, there is an $x$ such that we can modify the decision of $\mstrat$ in $s$ by $x$, not violate the soundness property and increase the penalty by at most $\varepsilon/|S|$. We thus need to pick $M$ such that $1/M \le x$. To finish the proof, we repeat this procedure for every state $s$.
\end{proof}

%
%
%
For the sake of completeness, we also show that \thmref{2_set_rand_subset}
does \emph{not} extend to dynamic penalties.
%
This is because, in this case, increasing the probability of allowing an action
can lead to an increased penalty if one of the successor states has a high expected penalty.
An example is shown in \figref{fig:chain_condition_cex}, for which we want to
reach the goal state $s_3$ with probability at least $0.5$.

\input{fig_counterexample}

\noindent
This implies $\mstrat(s_0,\{b\}){\cdot}\mstrat(s_1)(\{d\}){\ge} 0.5$, and
so $\mstrat(s_0)(\{b\}){>}0, \mstrat(s_1)(\{d\}){>}0$. If $\mstrat$ satisfies
the condition of \thmref{2_set_rand_subset}, then $\mstrat(s_0)(\{c\})=\mstrat(s_1)(\{e\})=0$, so an
opponent can always use $b$, forcing an expected penalty of
$\mstrat(s_0)(\{b\})+\mstrat(s_1)(\{d\})$, for a minimal value of $\sqrt{2}$.
However, the sound multi-strategy $\mstrat$ with
$\mstrat(s_0)(\{b\}){=}\mstrat(s_0)(\{c\}){=}0.5$ and $\mstrat(s_1,\{d\}){=}1$ achieves a
dynamic penalty of just $1$.

\subsection{Optimisations}\label{sec:optimisations}

We conclude this section on MILP-based multi-strategy synthesis by presenting some
optimisations that can be applied to our methods.
The general idea is to add additional constraints to the MILP problems
that will reduce the search space to be explored by a solver.
We present two different optimisations, targeting different aspects of our encodings:
(i) possible variable values; and (ii) penalty structures.

\startpara{Bounds on variable values}
%
In our encodings, for the variables $x_s$, we only specified very general lower and upper bounds that would
constrain their value. Narrowing down the set of values that a variable may take
can significantly reduce the search space and thus the solution time required by an MILP solver.
One possibility that works well in practice is to bound the values of the variables
by the minimal and maximal expected reward achievable from the given state, i.e.,
add the constraints:
\begin{align*}
\inf_{\sigma, \pi} E_{\sgame,s}^{\sigma,\pi}(r) & \le x_s \le \sup_{\sigma,\pi} E_{\sgame,s}^{\sigma,\pi}(r) &\text{for all $s\in S$}
\end{align*}
where both the infima and suprema above are constants obtained by applying standard
probabilistic model checking algorithms.

\startpara{Actions with zero penalties}
Our second optimisation exploits the case where an action has zero penalty assigned to it.
Intuitively, this action could always be disabled without harming the overall penalty of the multi-strategy.
On the other hand, enabling an action with zero penalty
might be the only way to satisfy the property and therefore we cannot disable all such actions.
However, it is enough to allow at most one action that has zero penalty.
For simplicity of the presentation, we assume $Z_s = \{ a \in A(s) \, | \, \penbase(s,a) = 0 \}$;
then formally we add the constraints
$\sum_{a\in Z_s} y_{s,a} \le 1$ for all $s\in S_\Diamond$.


\section{Experimental Results}\label{sec:expts}

We have implemented our techniques within PRISM-games~\cite{CFK+13},
an extension of the PRISM model checker~\cite{KNP11} for performing model checking and strategy synthesis on stochastic games.
PRISM-games can thus already be used for (classical) controller synthesis problems on
stochastic games. To this, we add the ability to synthesise multi-strategies
using the MILP-based method described in \sectref{sec:synthesis}.
%
%
Our implementation currently uses CPLEX~\cite{CPLEX} or Gurobi~\cite{GUROBI}
to solve MILP problems.
%
%
We investigated the applicability and performance of our approach
on a variety of case studies, some of which are existing benchmark examples
and some of which were developed for this work.
These are described in detail below and the files used can be found online~\cite{files}.
Our experiments were run on a PC with a 2.8GHz Xeon processor and 32GB of RAM, running Fedora 14.

\subsection{Deterministic Multi-strategy Synthesis}

We first discuss the generation of optimal \emph{deterministic} multi-strategies,
the results of which are presented in~\tabtabref{tab:perm_case_studies}{tab:perm_det_ms}.
\tabref{tab:perm_case_studies} summarises the models and properties considered.
For each model, we give:
the case study name, any parameters used,
the number of states ($|S|$) and of controller states ($|S_\Diamond|$), and the property used.
The final column gives, for comparison purposes,
the time required for performing classical (single) strategy synthesis on each model and property $\phi$.

In~\tabref{tab:perm_det_ms}, we show, for each different model,
the penalty value of the optimal multi-strategy and the time to generate it.
We report several different times, each for different combinations of
the optimisations described in \sectref{sec:optimisations}
(either no optimisations, one or both).
For the last result, we give times for both MILP solvers: CPLEX and Gurobi.

\startpara{Case studies}
Now, we move on to give further details for each case study,
illustrating the variety of ways that permissive controller synthesis can be used.
Subsequently, we will discuss the performance and scalability of our approach.

\begin{table}
\begin{center}	
\begin{tabular}{|c|c|c|c|c|c|}
\hline
\multirow{2}{2.3cm}{\centering Name\\{[parameters]}} &
\multirow{2}{1.8cm}{\centering Parameter\\values} & 				
\multirow{2}{1.2cm}{\centering States} &
\multirow{2}{2.0cm}{\centering Controller\\states}&
\multirow{2}{1.8cm}{\centering Property}&
\multirow{2}{2.4cm}{\centering Strat. synth. \\ time (s)}\\
&&&&&\\
\hline
\hline
\multirow{2}{2.1cm}{\centering \emph{cloud}\\$[vm]$}
 & 5 & 8,841 & 2,177 &$\calP_{\geq0.9999}[\,\future \mathit{deployed}\,]$&0.04\\ 
 & 6 & 34,953 & 8,705 & $\calP_{\geq0.999}[\,\future \mathit{deployed}\,]$&0.10 \\ 
\hline
\multirow{3}{2.1cm}{\centering \emph{android}\\$[r, s]$}
 & 1, 48 & 2,305 & 997 &\multirow{3}{*}{$\nprop{\mathit{time}}{10000}$}&0.16 \\ 
 & 2, 48 & 9,100 & 3,718 &&0.57 \\ 
 & 3, 48 & 23,137 & 9,025 && 0.93\\ 
\hline
\multirow{2}{2.1cm}{\centering \emph{mdsm}\\$[N]$}
 & 3 & 62,245 & 9,173 &$\calP_{\leq0.1}[\,\future \mathit{deviated}\,]$ &\multirow{2}{*}{5.64}\\ 
 & 3 & 62,245 & 9,173 &$\calP_{\leq0.01}[\,\future \mathit{deviated}\,]$ &\\ 
\hline
\multirow{2}{2.1cm}{\centering \emph{investor}\\$[vinit,vmax]$}
 & 5,10 & 10,868 & 3,344 &$\npropl{\mathit{profit}}{4.98}$ &0.87\\ 
 & 10, 15& 21,593 & 6,644 &$\npropl{\mathit{profit}}{8.99}$ &2.20\\ 
\hline
\multirow{2}{2.1cm}{\centering \emph{team-form}\\$[N]$}
 & 3 & 12,476 & 2,023 & \multirow{2}{*}{$\,\calP_{\geq0.9999}[\,\future \mathit{done}_1\,]\,$}&0.20 \\ 
 & 4 & 96,666 & 13,793&&0.83\\ 
\hline
\emph{cdmsn} $[N]$
 & 3 & 1240 & 604 &$\calP_{\geq0.9999}[\,\future \mathit{prefer}_1]$&0.05 \\ 
\hline
\end{tabular}
\end{center}
\caption{Details of the models and properties used for experiments with deterministic multi-strategies, and execution times for \emph{single} strategy synthesis.}
\label{tab:perm_case_studies}
\end{table}

\begin{table}
\begin{center}	
\begin{tabular}{|c|c|c|c|c|c|c|c|c|}
\hline
\multirow{3}{2.3cm}{\centering Name\\{[parameters]}} &
\multirow{3}{1.8cm}{\centering Parameter\\values} & 				
\multirow{3}{1.4cm}{\centering Penalty} &
\multicolumn{5}{c|}{\centering Multi-strategy synthesis time (s)}\\
\cline{4-8}
&&&\multicolumn{4}{c|}{\centering CPLEX} & Gurobi \\
\cline{4-8}
&&&No-opts&Bounds&Zero&Both&Both\\
\hline
\hline
\multirow{2}{2.1cm}{\centering \emph{cloud}\\$[vm]$}
 & 5 & 0.001 & 2.5 & 3.35 & 13.04 & 10.36 & 1.45 \\ 
 & 6 & 0.01 & 62.45 & $^*$ & 63.59 & 25.37 & 4.73 \\ 
\hline
\multirow{3}{2.1cm}{\centering \emph{android}\\$[r, s]$}
 & 1, 48 & 0.0009 & 1.07 & 0.66 & 1.04 & 0.48 & 0.53 \\ 
 & 2, 48 & 0.0011 & 28.56 & 8.41 & 28.48 & 8.42 & 3.6 \\ 
 & 3, 48 & 0.0013 & $^*$ & 13.41 & $^*$ & 13.30 & 47.62 \\ 
\hline
\multirow{2}{2.1cm}{\centering \emph{mdsm}\\$[N]$}
 & 3 & 52 & 28.06 & 36.28 & 27.88 & 33.72 & 19.40 \\ 
 & 3 & 186 & 11.89 & 11.57 & 11.88 & 11.56 & 12.27 \\ 
\hline
\multirow{2}{2.1cm}{\centering \emph{investor}\\$[vinit,vmax]$}
 & 5,10 & 1 & 68.64 & 131.38 & 68.90 & 131.36 & 12.02  \\ 
 & 10, 15 & 1 & $^*$ & $^*$ & $^*$ & $^*$ & 208.95 \\ 
\hline
\multirow{2}{2.1cm}{\centering \emph{team-form}\\$[N]$}
 & 3 & 0.890 & 0.15 & 0.26 & 0.15 & 0.26 & 0.80 \\ 
 & 4 & 0.890 & 249.36 & 249.49 & 186.41 & 184.50 & 3.84 \\ 
\hline
\emph{cdmsn} $[N]$
 & 3 & 2 & 0.57 & 0.62 & 0.62 & 0.61 & 1.65 \\ 
\hline
\end{tabular}
\\[0.4em]
{\scriptsize * No optimal solution to MILP problem within 5 minute time-out.}
\end{center}
\vspace{-0.2cm}
\caption{Experimental results for synthesising optimal deterministic multi-strategies.}
\label{tab:perm_det_ms}
\end{table}

%

\vskip3pt\noindent{\em cloud:}
We adapt a PRISM model from~\cite{CJK12} to synthesise
deployment of services across virtual machines (VMs) in a cloud infrastructure.
Our property $\phi$ specifies that, with high probability, 
services are deployed to a preferred subset of VMs,
and we then assign unit (dynamic) penalties to all actions corresponding to deployment on this subset.
The resulting multi-strategy has very low expected penalty (see \tabref{tab:perm_det_ms})
indicating that the goal $\phi$ can be achieved whilst the controller experiences
reduced flexibility only on executions with low probability.

\vskip3pt\noindent{\em android:}
We apply permissive controller synthesis to a model created for
runtime control of an Android application that provides real-time stock monitoring.
We extend the application to use multiple data sources
and synthesise a multi-strategy which specifies an
efficient runtime selection of data sources
($\phi$ bounds the total expected response time).
We use static penalties, assigning higher values to
actions that select the two most efficient data sources at each time point
and synthesise a multi-strategy that always provides a choice of at least two sources
(in case one becomes unavailable), while preserving $\phi$.

\vskip3pt\noindent{\em mdsm:}
Microgrid demand-side management (MDSM) is a randomised scheme for managing local energy usage.
A stochastic game analysis~\cite{CFK+12} previously showed it is beneficial
for users to selfishly deviate from the protocol,
by ignoring a random back-off mechanism designed to reduce load at busy times.
We synthesise a multi-strategy for a (potentially selfish) user,
with the goal $(\phi$) of bounding the probability of deviation (at either $0.1$ or $0.01$).
The resulting multi-strategy could be used to modify the protocol, restricting the behaviour of this user to reduce selfish behaviour.
To make the multi-strategy as permissive as possible, restrictions are only introduced where necessary to ensure $\phi$.
We also guide where restrictions are made by assigning (static) penalties at certain times of the day.



\vskip3pt\noindent{\em investor:}
This example \cite{MM07} synthesises strategies for a futures market investor,
who chooses when to reserve shares, operating in a
(malicious) market which can periodically ban him/her from investing.
We generate a multi-strategy that achieves $90\%$ of the maximum expected profit
(obtainable by a single strategy) and assign (static) unit penalties to all actions,
showing that, after an immediate share purchase,
the investor can choose his/her actions freely and still meet the $90\%$ target.

\vskip3pt\noindent{\em team-form:}
This example~\cite{CKPS11} synthesises strategies for forming teams of agents
in order to complete a set of collaborative tasks.
Our goal ($\phi$) is to guarantee that a particular task is completed with high probability (0.9999).
We use (dynamic) unit penalties on all actions of the first agent
and synthesise a multi-strategy representing several possibilities for this agent while still achieving the goal.

\vskip3pt\noindent{\em cdmsn:}
Lastly, we apply permissive controller synthesis
to a model of a protocol for collective decision making in sensor networks (CDMSN)~\cite{CFK+12}.
We synthesise strategies for nodes in the network
such that consensus is achieved with high probability (0.9999).
We use (static) penalties inversely proportional to the energy associated with each action a node can perform
to ensure that the multi-strategy favours more efficient solutions.


\startpara{Performance and scalability}
Unsurprisingly, permissive controller synthesis is more costly to execute than
classical controller synthesis -- this is clearly seen by comparing the times in the rightmost column of
\tabref{tab:perm_case_studies} with the times in \tabref{tab:perm_det_ms}.
However, we successfully synthesised deterministic multi-strategies for a wide range of models and properties,
with model sizes ranging up to approximately 100,000 states.
The performance and scalability of our method is affected (as usual) by the state space size.
In particular, it is also affected by the number of actions in controller states,
since these result in integer MILP variables, which are the most expensive part of the solution.

The performance optimisations presented in~\sectref{sec:optimisations} often allowed
us to obtain an optimal multi-strategy quicker. In many cases, 
it proved beneficial to apply both optimisations at the same time.
In the best case ($\mathit{android},r{=}3,s{=}48),$ an order of magnitude gain was observed.
We reported a slowdown after applying optimisation in the case of \emph{investor}.
We attribute this to the fact that adding bounds on variable value can make finding
the initial solution of the MILP problem harder, causing a slowdown of the overall solution process.

Both solvers were run using their off-the-shelf configuration and Gurobi proved to be a more efficient solver.
In the case of CPLEX, we also observed worse numerical stability, causing it to return a sub-optimal
multi-strategy as optimal. In the case of Gurobi, we did not see any such behaviour.

\subsection{Randomised multi-strategy synthesis}
Next, we report the results for approximating optimal \emph{randomised} multi-strategies.
\tabref{tab:perm_case_studies_rand} summarises the models and properties used.
In~\tabref{tab:perm_rand_ms_states}, we report the effects on state space size
caused by adding the approximation gadgets to the games.
We picked three different granularities $M=100, M=200$ and $M=300$; for higher values of $M$ we did
not observe improvements in the penalties of the generated multi-strategies.
Finally, in~\tabref{tab:perm_rand_ms},
we show the penalties obtained by the randomised multi-strategies that were generated.
We compare the (static) penalty value of the randomised multi-strategies to the value
obtained by optimal deterministic multi-strategies for the same models.
The MILP encodings for randomised multi-strategies are larger than deterministic ones and thus slower to solve,
so we impose a time-out of 5 minutes. We used Gurobi as the MILP solver in every case.

\begin{table}
\begin{center}	
\begin{tabular}{|c|c|c|c|c|}
\hline
\multirow{2}{2.6cm}{\centering Name\\{[parameters]}} &
\multirow{2}{2.0cm}{\centering Parameters\\values} &
\multirow{2}{1.4cm}{\centering States} &
\multirow{2}{1.8cm}{\centering Controller\\states}&
\multirow{2}{1.2cm}{\centering Property}\\
&&&&\\
\hline
\hline
\multirow{2}{1.6cm}{\centering \emph{android}\\$[r,s]$}
& 1,1 & 49 & 10 & $\calP_{\geq0.9999}[\,\future\mathit{done}]$\\
& 1,10 & 481 & 112 & $\calP_{\geq0.999}[\,\future\mathit{done}]$\\
\hline
\multirow{2}{1.6cm}{\centering \emph{cloud}\\$[vm]$}
& \multirow{2}{*}{5} & \multirow{2}{*}{8,841} & \multirow{2}{*}{2,177} &\multirow{2}{*}{$\,\calP_{\geq0.9999}[\,\future\mathit{deployed}]\,$}\\
&&&&\\
\hline
\multirow{1}{2.6cm}{\centering \emph{investor}\\$[vinit,vmax]$}
& \multirow{2}{*}{5,10} & \multirow{2}{*}{10,868} & \multirow{2}{*}{3,344} &\multirow{2}{*}{$\npropl{\mathit{profit}}{4.98}$}\\
&&&&\\
\hline
\multirow{2}{2.2cm}{\centering \emph{team-form}\\$[N]$}
& \multirow{2}{*}{3} & \multirow{2}{*}{12,476} & \multirow{2}{*}{2,023} &\multirow{2}{*}{$\calP_{\geq0.9999}[\,\future\mathit{done}_1\,]$}\\
&&&&\\
\hline
\multirow{2}{1.6cm}{\centering \emph{cdmsn} $[N]$}
 & \multirow{2}{*}{3} & \multirow{2}{*}{1240} & \multirow{2}{*}{604} &\multirow{2}{*}{$\calP_{\geq0.9999}[\,\future\mathit{prefer}_1]$}\\
&&&&\\
\hline
\end{tabular}
\\
\end{center}
\caption{Details of models and properties for approximating optimal randomised multi-strategies.}
\label{tab:perm_case_studies_rand}
\end{table}

We are able to generate a sound multi-strategy for all the examples;
in some cases it is optimally permissive, in others it is not (denoted by a $^*$ in \tabref{tab:perm_rand_ms}).
As would be expected, often, larger values of $M$ give smaller penalties.
In some cases, this is not true, which
we attribute to the size of the MILP problem (which grows with $M$).
For all examples, we built
randomised multi-strategies with smaller or equal penalties than the deterministic ones.

\begin{table}
\begin{center}	
\begin{tabular}{|c|c|c|c||c|c|c|}
\hline
\multirow{2}{2.6cm}{\centering Name\\{[parameters]}} &
\multirow{2}{2.0cm}{\centering Parameters\\values} &
\multirow{2}{1.4cm}{\centering States} &
\multirow{2}{1.8cm}{\centering Controller\\states}&
\multicolumn{3}{|c|}{States} \\
\cline{5-7}
&&&&$\gran{=}100$&$\gran{=}200$&$\gran{=}300$ \\
\hline
\hline
\multirow{2}{2.0cm}{\centering \emph{android}  \\$[r,s]$}
& 1,1 & 49 & 10 & 90 & 94 & 98\\
& 1,10 & 481 & 112 & 1629 & 1741 & 1853\\
\hline
\multirow{2}{2.0cm}{\centering \emph{cloud} $[vm]$}
& \multirow{2}{*}{5} & \multirow{2}{*}{8,841} & \multirow{2}{*}{2,177} & \multirow{2}{*}{29447} & \multirow{2}{*}{32686} & \multirow{2}{*}{35233}\\
&&&&&&\\
\hline
\multirow{2}{2.6cm}{\centering \emph{investor}   \,$[vinit,vmax]$}
& \multirow{2}{*}{5,10} & \multirow{2}{*}{10,868} & \multirow{2}{*}{3,344} & \multirow{2}{*}{33440} & \multirow{2}{*}{35948} & \multirow{2}{*}{38456}\\
&&&&&&\\
\hline
\multirow{2}{2.2cm}{\centering \emph{team-form} $[N]$}
& \multirow{2}{*}{3} & \multirow{2}{*}{12,476} & \multirow{2}{*}{2,023} & \multirow{2}{*}{31928} & \multirow{2}{*}{33716} & \multirow{2}{*}{35504}\\
&&&&&&\\
\hline
\multirow{2}{2.0cm}{\centering \emph{cdmsn} $[N]$}
& \multirow{2}{*}{3} & \multirow{2}{*}{1240} & \multirow{2}{*}{604} & \multirow{2}{*}{3625} & \multirow{2}{*}{3890} & \multirow{2}{*}{4155}\\
&&&&&&\\
 \hline
\end{tabular}
\\
\end{center}
\caption{State space growth for approximating optimal randomised multi-strategies.}
\label{tab:perm_rand_ms_states}
\end{table}

\begin{table}
\begin{center}	
\begin{tabular}{|c|c|c|c||c|c|c|c|}
\hline
\multirow{2}{2.2cm}{\centering Name\\{[parameters]}} &
\multirow{2}{1.9cm}{\centering Parameters\\values} &
\multirow{2}{1.4cm}{\centering States} &
\multirow{2}{1.7cm}{\centering Controller\\states}&
\multirow{2}{0.8cm}{\centering Pen. (det.)}&
\multicolumn{3}{|c|}{Penalty (randomised)} \\
\cline{6-8}
&&&&&$\gran{=}100$&$\gran{=}200$&$\gran{=}300$ \\
\hline
\hline
\multirow{2}{2.0cm}{\centering \emph{android} \\$[r,s]$}
& 1,1 & 49 & 10 & 1.01 & 0.91 & 0.905 & 0.903\\
& 1,10 & 481 & 112 & 19.13 & 12.27$^*$ & 9.12$^*$ & 17.18$^*$\\
\hline
\multirow{2}{2.0cm}{\centering \emph{cloud}\\$[vm]$}
& \multirow{2}{*}{5} & \multirow{2}{*}{8,841} & \multirow{2}{*}{2,177} &\multirow{2}{*}{1} & \multirow{2}{*}{0.91$^*$} & \multirow{2}{*}{0.905$^*$} & \multirow{2}{*}{0.91$^*$}\\
&&&&&&&\\
\hline
\multirow{2}{2.2cm}{\centering \emph{investor}\\ $[vinit,vmax]$}
& \multirow{2}{*}{5,10} & \multirow{2}{*}{10,868} & \multirow{2}{*}{3,344} &\multirow{2}{*}{1}& \multirow{2}{*}{1$^*$} & \multirow{2}{*}{1$^*$} & \multirow{2}{*}{1$^*$}\\
&&&&&&&\\
\hline
\multirow{2}{2.0cm}{\centering \emph{team-form}  $[N]$}
& \multirow{2}{*}{3} & \multirow{2}{*}{12,476} & \multirow{2}{*}{2,023} & \multirow{2}{*}{264} & \multirow{2}{*}{263.96$^*$} & \multirow{2}{*}{263.95$^*$} & \multirow{2}{*}{263.95$^*$}\\
&&&&&&&\\
\hline
\multirow{2}{2.0cm}{\centering \emph{cdmsn} $[N]$}
 & \multirow{2}{*}{3}& \multirow{2}{*}{1240}& \multirow{2}{*}{604}&\multirow{2}{*}{2}&\multirow{2}{*}{0.38$^*$}&\multirow{2}{*}{1.9$^*$}&\multirow{2}{*}{0.5$^*$}\\
&&&&&&&\\
\hline
\end{tabular}
\\
[0.3em]
{\scriptsize $^*$ Sound but possibly non-optimal multi-strategy obtained after 5 minute MILP time-out.}
\end{center}
\caption{Experimental results for approximating optimal randomised multi-strategies.}
\label{tab:perm_rand_ms}
\end{table}

\section{Conclusions}

We have presented a framework for permissive controller synthesis on stochastic two-player games,
based on generation of multi-strategies that guarantee a specified objective
and are optimally permissive with respect to a penalty function.
We proved several key properties, developed MILP-based synthesis methods
and evaluated them on a set of case studies.

In this paper, we have imposed several restrictions on permissive controller synthesis.
Firstly, we focused on properties expressed in terms of expected total reward
(which also subsumes the case of probabilistic reachability).
A possible topic for future work would be to consider more expressive temporal logics or parity objectives.
The results might also be generalised so that  both positive and negative rewards can be used,
for example by using the techniques of~\cite{TV87}.
We also restricted our attention to memoryless multi-strategies,
rather than the more general class of history-dependent multi-strategies.
Extending our theory to the latter case and exploring the additional power
brought by history-dependent multi-strategies is another interesting direction of future work.


\bibliographystyle{abbrv}
\bibliography{main}

\appendix

\newpage
\section{Appendix}

\subsection{Proof of \thmref{hardness} (NP-hardness)}\label{appx:md_hardness}

%
%
We start with the case of randomised multi-strategies and static penalties which is the most delicate. Then we
analyse the case of randomised multi-strategies and dynamic penalties, and finally show that this case can easily be modified for the remaining two combinations.

\startpara{Randomised multi-strategies and static penalties}
We give a reduction from the Knapsack problem.
Let $n$ be the number of items, each of which can
either be or not be put in a knapsack, let $v_i$ and $w_i$ be the value and the weight of item $i$, respectively,
and let $V$ and $W$ be the bounds on the value and weight of the items to be picked. We assume that $v_i\le 1$ for every $1\le i \le n$,
and that all numbers $v_i$ and $w_i$ are given as fractions with denominator $q$.
Let us construct the following MDP, where $m$ is chosen such that
$\smash{2^{-m} < \frac{1}{q}}$ and $\smash{2^{-m}\cdot W \le \frac{1}{q}}$.
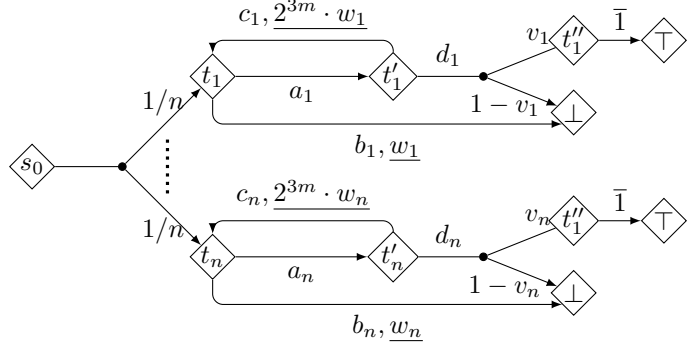
\begin{center}
\input{fig_hard_rand1}
\end{center}
The rewards and penalties are
as given by the overlined and underlined expressions, and set to $0$ where not present. In particular, note that for any state $s$ different from $\top$
the probability of reaching $\top$ from $s$ is the same as the expected total reward from $s$.

We show that there is a multi-strategy $\mstrat$ sound for the property
$\smash{\property{r}{V/n}}$ 
such that $\pendef{\penbase}{\sta}{\mstrat} \le W + 2^{-m}\cdot W$
if and only if the answer to the Knapsack problem is ``yes''.

In the direction $\Leftarrow$, let $I\subseteq\{1,\ldots,n\}$
be the set of items put in the knapsack. It suffices to define the multi-strategy $\mstrat$ by:
\begin{itemize}
 \item $\mstrat(t'_i)(\{c_i,d_i\})=1-2^{-4m}$, $\mstrat(t'_i)(\{d_i\}) = 2^{-4m}$, $\mstrat(t_i)(\{a_i\}) = 1$ for $i\in I$,
 \item $\mstrat(t'_i)(\{c_i,d_i\})=1$, $\mstrat(t_1)(\{a_i,b_i\}) = 1$ for $i\not\in I$.
\end{itemize}

In the direction $\Rightarrow$, let us have a multi-strategy $\mstrat$ satisfying the assumptions.
Let $P(s\rightarrow s')$ denote the lower bound on the probability of reaching $s'$ from $s$ under a
strategy which complies with the multi-strategy $\mstrat$.
Denote by $I\subseteq \{1,\ldots,n\}$ the indices $i$ such that $P(t_i \rightarrow \top) \ge 2^{-m}$.

Let $\beta_i = \mstrat(t_i)(\{a_i\})$ and $\alpha_i = \mstrat(t'_1)(\{d_i\})$. We will show that for
any $i\in I$ we have $\alpha_i \ge 2^{-m}(1-\beta_i)$. When $\beta_i=1$, this obviously holds. For $\beta_i < 1$,
assume that $\alpha_i \ge 2^{-m}(1-\beta_i)$.
Because the optimal strategy $\sigma$ will pick $b_i$ and $c_i$ whenever they are available, we have:
\begin{align*}
 P(t_i \rightarrow \top) &= \beta_i\cdot \sum_{j=0}^{\infty}((1-\alpha_i)\cdot\beta_i)^j\cdot \alpha_i \cdot v_i
     = \frac{\alpha_i\beta_i v_i}{1- (1-\alpha_i)\beta_i}
     = \frac{\alpha_i\beta_i v_i}{1-\beta_i + \alpha_i\beta_i}\\
     &< \frac{\alpha_i\beta_i}{1-\beta_i + \alpha_i\beta_i}
     < \frac{2^{-m}(1-\beta_i)\beta_i}{1-\beta_i + 2^{-m}(1-\beta_i)\beta_i}
     < \frac{2^{-m}\beta_i}{1+ 2^{-m}\beta_i}
     \le 2^{-m}
\end{align*}
which is a contradiction with $i\in I$.
Hence, $\alpha_i \ge 2^{-m}(1-\beta_i)$ and so:
\[
 \pendef{\penbase}{\sta}{\mstrat, t_i} + \pendef{\penbase}{\sta}{\mstrat, t'_i} = \beta_i w_i + \alpha_i 2^{3m} w_i \ge \beta_i w_i + 2^{-m} (1-\beta_i) 2^{3m} w_i \ge w_i
\]
We have:
\[
 \sum_{i\in I} w_i \le \sum_{i\in I} \big(\pendef{\penbase}{\sta}{\mstrat, t_i} + \pendef{\penbase}{\sta}{\mstrat, t'_i}\big) \le W + 2^{-m}\cdot W
\]
and, because $\sum_{i\in I} w_i$ and $W$ are fractions with denominator $q$, by the choice of $m$, we can infer that $\sum_{i\in I} w_i \le W$.
Similarly:
\[
 \sum_{i\in I}\frac{1}{n} v_i \ge \sum_{i\in I} \frac{1}{n} P(t_i \rightarrow \top)
  \ge \Big(\frac{1}{n}\sum_{i=1}^n P(t_i\rightarrow \top)\Big) - \frac{1}{n}2^{-m}n
  \ge \frac{1}{n}V - 2^{-m}
\]
and again, because $\sum_{i\in I} v_i$ and $V$ are
fractions with denominator $q$, by the choice of $m$ we can infer that $\sum_{i\in I} v_i \ge V$.
Hence, in the instance of the knapsack problem, it suffices to pick exactly items from $I$ to satisfy the restrictions.


\startpara{Randomised multi-strategies with dynamic penalties}
The proof is analogous to the proof above,
we only need to modify the MDP and the computations.
For an instance of the Knapsack problem given as before, we construct the following MDP:
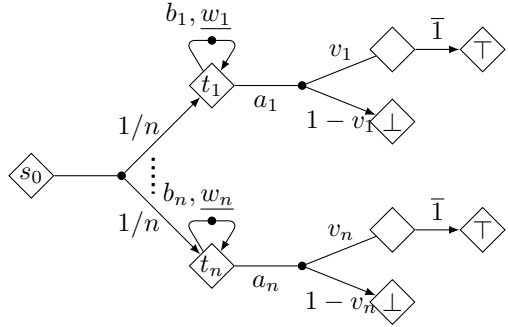
\begin{center}
\input{fig_hard_rand2}
\end{center}
We claim that there is a multi-strategy $\mstrat$ sound for the property
$\smash{\property{r}{V/n}}$ 
such that $\pendef{\penbase}{\dyn}{\mstrat} \le \frac{1}{n}W$
if and only if the answer to the Knapsack problem is ``yes''.

In the direction $\Leftarrow$, for $I\subseteq\{1,\ldots,n\}$
the set of items in the knapsack, we define $\mstrat$ by $\mstrat(t_i)(\{a_i\})=1$ for $i\in I$ and by allowing
all actions in every other state.

In the direction $\Rightarrow$,
let us have a multi-strategy $\mstrat$ satisfying the assumptions.
Let $P(s\rightarrow s')$ denote the lower bound on the probability of reaching $s'$ from $s$ under a
strategy which complies with the multi-strategy $\mstrat$.
Denote by $I\subseteq \{1,\ldots,n\}$ the indices $i$ such that $\mstrat(t_i)(\{a_i\})>0$.
Observe that $P(t_i \rightarrow \top) = v_i$ if $i\in I$ and $P(t_i \rightarrow \top) = 0$ otherwise.
Hence:
\[
 \sum_{i\in I}\frac{1}{n} v_i = \sum_{i\in I} \frac{1}{n} P(t_i \rightarrow \top)
  = \frac{1}{n}\sum_{i=1}^n P(t_i\rightarrow \top)
  \ge \frac{1}{n}V
\]
and for the penalty, denoting $x_i:= \mstrat(t_i)(\{a_i\})$, we get:
\begin{equation}\label{hardness_eqn_pen}
  \frac{1}{n}W \ge \pendef{\penbase}{\dyn}{\mstrat}
  =\frac{1}{n} \sum_{i=0}^n \sum_{j=0}^\infty (1-x_i)^j x_i w_i
  =\frac{1}{n} \sum_{i\in I} \sum_{j=0}^\infty (1-x_i)^j x_i w_i
  =\frac{1}{n} \sum_{i\in I} w_i
\end{equation}
because the strategy that maximises the penalty will pick $b_i$ whenever it is available.
Hence, in the instance of the knapsack problem, it suffices to pick exactly items from $I$ to satisfy the restrictions.

\startpara{Deterministic multi-strategies and dynamic penalties}
The proof is identical to the proof for randomised multi-strategies and dynamic penalties above: observe that the multi-strategy
constructed there from an instance of Knapsack is in fact deterministic.

\startpara{Deterministic multi-strategies and static penalties} 
The proof is obtained by a small modification of the proof for randomised multi-strategies and dynamic penalties above.
Instead of requiring $\pendef{\penbase}{\dyn}{\mstrat} \le \frac{1}{n}W$, we require $\pendef{\penbase}{\sta}{\mstrat} \le W$
and (\ref{hardness_eqn_pen}) changes to:
\begin{equation*}
  W \ge \pendef{\penbase}{\sta}{\mstrat}
  = \sum_{i=0}^n x_i w_i
  = \sum_{i\in I} w_i \, .
\end{equation*}

\subsection{Proof of \thmref{upper-bound} (Upper Bounds)}\label{appx:upper-bound}
%
%
%
We consider the two cases of deterministic and randomised multi-strategies separately,
showing that they are in NP and PSPACE, respectively.
To simplify readability, the proofs make use of constructions
that appear later in the main paper than \thmref{upper-bound}.
Note that those constructions do not build on the theorem and
so there is no cyclic dependency.

\startpara{Deterministic multi-strategies}
%
%
For deterministic multi-strategies, it suffices to observe that the problem can be solved by
verifying an MILP instance constructed in polynomial time
(and, additionally, in the case of dynamic penalties: a polynomial-time identification of the infinite-penalty case --
see \thmref{thm:main-static} and \thmref{thm:main-dynamic}). Since the problem of solving an MILP instance
is in NP, the result follows.

\startpara{Randomised multi-strategies}
We now show that the permissive controller synthesis problem is in
PSPACE for randomised multi-strategies and static penalties.
The proof for dynamic penalties is similar.

The proof proceeds by constructing a polynomial-size closed formula $\Psi$ of
the existential fragment of $(\Rset, +, \cdot, \leq)$
such that $\Psi$ is true if and only if there is a multi-strategy ensuring the required penalty and reward.
Because determining the validity of a closed formula of the existential fragment of
$(\Rset, +, \cdot, \leq)$ is in PSPACE~\cite{Canny}, we obtain the desired result.

We do not construct the formula $\Psi$ explicitly, but only sketch the main idea.
Recall that in \sectref{sec:approximation} we presented a reduction that allows us to
\emph{approximate} the existence of a multi-strategy using the construction described
on page \pageref{sec:approximation} and in \figref{fig:approximation}. Note that if we
\emph{knew} the probabilities with which the required multi-strategy $\mstrat$ chooses
some sets in a state $s$, we could use these probabilities instead of the numbers $p_1,\ldots,p_m$
in \figref{fig:approximation}. In fact, by \thmref{2_set_rand_chain_condition} we
would only need $n=2$, i.e. two numbers per state. Now knowing these numbers $p_1^s, p_2^s$ for each state,
we can construct, in polynomial time a polynomial-size instance (disregarding the size of representation of the numbers
$p_1^s,p_2^s$) of an MILP problem such that the optimal solution for the problem
is the optimal reward/penalty under the multi-strategy $\mstrat$. Of course, we do \emph{not} know the numbers $p_1^s,p_2^s$
a priori, and so we cannot use MILP for the solution. Instead, we treat those numbers as variables in
the formula $\Psi$ which is build from the constraints in the MILP problem, by adding existential quantification for $p_1^s,p_2^s$ for all $s\in S$
and all other variables from the MILP problem,
requiring $p_1^s,p_2^s$ to be from $(0,1)$ and $p_1^s = 1-p_2^s$, and adding the restriction on the total reward.
Note that in the case of dynamic penalties we need to treat optimal infinite penalty separately, similarly to the
construction on page~\pageref{page:inf-penalty}.

\subsection{Proof of \thmref{square-root-sum} (Square-root-sum Reduction)}
\label{appx:reduction}
%
%
Let $x_1$,\dots,$x_n$ and $y$ be numbers giving the instance of the square-root-sum problem, i.e. we aim to determine
whether $\sum_{i = 1}^n \sqrt{x_i} \le y$. We construct the game from \figref{fig:sqrt}.

\begin{figure}
\input{fig_sqrt}
\caption{The game for the proof of \thmref{square-root-sum}.}
\label{fig:sqrt}
\end{figure}

The penalties are as given by the underlined numbers, and the rewards $1/x_i$ are awarded under the actions 
$b_i$.

\startpara{Static penalties}
We first give the proof for static penalties.
We claim that there is a multi-strategy $\mstrat$ sound for the property $\property{r}{1}$ such that
$\pendef{\penbase}{\sta}{\mstrat} \le 2\cdot y$ if and only if $\sum_{i = 1}^n \sqrt{x_i} \le y$.

In the direction $\Leftarrow$ let us define a multi-strategy
$\theta$ by
$\theta(t'_i)(\{c'_i\})=\theta(\bar t_i)(\{\bar c_i\})=\sqrt{x_i}$
and
$\theta(t'_i)(\{a'_i,c'_i\})=\theta(\bar t_i)(\{\bar a_i, \bar c_i\})=1-\sqrt{x_i}$, and allowing all actions in all remaining states.
We then have:
$\pendef{\penbase}{\sta}{\mstrat}=\sum_{i=1}^n 2\cdot \sqrt{x_i}$
and the reward achieved is: 
\[
 \frac{1}{n}\sum_{i=1}^n \min\{x_i\cdot \frac{1}{x_i}, \sqrt{x_i}\cdot\sqrt{x_i}\frac{1}{x_i}\}\quad = \quad 1.
\]

In the direction $\Rightarrow$, let $\theta$ be an arbitrary multi-strategy sound for the property
$\property{r}{1}$ satisfying $\pendef{\penbase}{\sta}{\mstrat} \le 2\cdot y$ .
Let $z'_i = \theta(t'_i)(\{c'_i\})$ and
$\bar z_i = \theta(\bar t_i)(\{\bar c_i\})$.
The reward achieved is:
\[
 \frac{1}{n} \sum_{i=1}^n \min\{x_i\cdot \frac{1}{x_i}, z'_i\cdot\bar z_i \frac{1}{x_i}\}
 = \frac{1}{n} \sum_{i=1}^n \min\{1, z'_i\cdot\bar z_i \frac{1}{x_i}\}
\]
which is greater or equal to $1$ if and only if $z'_i\cdot\bar z_i \ge x_i$ for every $i$.
We show that $z'_i + \bar z_i \ge 2\cdot \sqrt{x_i}$, by analysing the possible cases: If both $z'_i$ and $\bar z_i$ are greater than
$\sqrt{x_i}$, we are done. The case $z'_i, \bar z_i < \sqrt{x_i}$ cannot take place. As for the remaining case, w.l.o.g., suppose that
$z'_i = \sqrt{x_i} + p$ and $\bar z_i = \sqrt{x_i} - q$ for some non-negative $p$ and $q$. Then  
$x_i \le (\sqrt{x_i} + p) \cdot (\sqrt{x_i} - q) = x_i + (p-q) \sqrt{x_i} - pq$, and for this to be at least $x_i$ we necessarily have
$p\ge q$, and so $z'_i + \bar z_i = \sqrt{x_i} + p + \sqrt{x_i} - q \ge 2\cdot \sqrt{x_i}$.

Hence, we get $
 \sum_{i=1}^n 2\cdot \sqrt{x_i} \le \sum_{i=1}^n \big( z'_i + \bar z_i \big)  = \pendef{\penbase}{\sta}{\mstrat} \le 2\cdot y$.

\startpara{Dynamic penalties}
We now proceed with dynamic penalties, where the analysis is similar. Let us use the same game as before, but
in addition assume that the penalty assigned to actions $c'_i$ and $\bar c'_i$ is equal to $1$.
We claim that there is a multi-strategy $\mstrat$ sound for the property $\property{r}{1}$ such that
$\pendef{\penbase}{\dyn}{\mstrat} \le 2\cdot y/n$ if and only if $\sum_{i = 1}^n \sqrt{x_i} \le y$.

In the direction $\Leftarrow$ let us define a multi-strategy $\theta$ as before, and obtain
$\pendef{\penbase}{\dyn}{\mstrat}=\frac{1}{n}\sum_{i=1}^n 2\cdot \sqrt{y_i}$.

In the direction $\Rightarrow$, let $\theta$ be an arbitrary multi-strategy sound for the property
$\property{r}{1}$ satisfying $\pendef{\penbase}{\dyn}{\mstrat} \le 2\cdot y/n$ .
Let $z'_i = \theta(t'_i)(\{c'_i\})$,
$\bar z_i = \theta(\bar t_i)(\{\bar c_i\})$,
$u'_i = \theta(t'_i)(\{a'_i\})$, and
$\bar u_i = \theta(\bar t_i)(\{\bar a_i\})$.

Exactly as before we show that $z'_i + \bar z_i \ge 2\cdot \sqrt{x_i}$,
and so:
\begin{multline*}
 \frac{1}{n}\sum_{i=1}^n 2\cdot \sqrt{x_i} \le \frac{1}{n}\sum_{i=1}^n \big(z'_i + \bar z_i \big) \le
 \frac{1}{n}\sum_{i=1}^n \big((z'_i + u'_i) + (1- u'_i)\cdot (\bar z_i + \bar u_i)\big)\\ 
 = \pendef{\penbase}{\dyn}{\mstrat} \le 2\cdot y /n.
\end{multline*}

\end{document}

%% file: macros-dave.tex
\newcommand{\property}[2]{\calR^{#1}_{\,\geq{#2}}[\,\mathtt{C}\,]} 
\newcommand{\propbt}[2]{\calR^{#1}_{\,\bowtie{#2}}[\,\mathtt{C}\,]} 
\newcommand{\nprop}[2]{\calR^{#1}_{\,\leq{#2}}[\,\mathtt{C}\,]} 
\newcommand{\npropl}[2]{\calR^{#1}_{\,\geq{#2}}[\,\mathtt{C}\,]} 

\newcommand{\strat}{\sigma}
\newcommand{\strats}{\Sigma}
\newcommand{\mstrat}{\theta}
\newcommand{\mstrats}{\Theta}
\newcommand{\Compl}{\triangleleft}
\newcommand{\dete}{\mathit{det}}
\newcommand{\rand}{\mathit{rand}}
\newcommand{\ctrlstrat}{\sigma}
\newcommand{\envstrat}{\pi}
\newcommand{\penrew}[1]{\psi^{#1}_{\mathit{rew}}}

\newcommand{\penbase}{\psi}
\newcommand{\pentype}{\tau}
\newcommand{\pendef}[3]{\mathit{pen}_{#2}(#1,#3)}
\newcommand{\sta}{\mathit{sta}}
\newcommand{\dyn}{\mathit{dyn}}

\newcommand{\eitherplayer}{\circ}
\newcommand{\gran}{M}
\newcommand{\penmax}{\mathit{pen}_{\max}}

\newcommand{\support}[1]{\mathit{supp}(#1)}

\newcommand{\enab}[1]{A(#1)}

\newcommand{\prmdp}[2]{\M{\otimes}\A}

\newcommand{\appref}[1]{Appx.~\ref{#1}}
\newcommand{\sectref}[1]{Section~\ref{#1}}

\newcommand{\figref}[1]{Fig.~\ref{#1}}
\newcommand{\tabref}[1]{Tab.~\ref{#1}}
\newcommand{\egref}[1]{Ex.~\ref{#1}}
\newcommand{\eqnref}[1]{(\ref{#1})}
\newcommand{\thmref}[1]{Theorem~\ref{#1}}

\newcommand{\lemref}[1]{Lemma~\ref{#1}} 
\newcommand{\defref}[1]{Defn.~\ref{#1}}

\newcommand{\figfigref}[2]{Fig.s~\ref{#1} and \ref{#2}}
\newcommand{\tabtabref}[2]{Tab.s~\ref{#1} and \ref{#2}}

\newcounter{exampcount}
{\vskip6pt}

\newcommand{\startpara}[1]{{%
\vskip6pt\noindent
{\bf #1.}}}

\def\ra{{\rightarrow}}

\def\Nset{\mathbb{N}}

\def\Qset{\mathbb{Q}}
\def\Rset{\mathbb{R}}

\def\ra{\rightarrow} 
\def\rmdef{=}

\def\squareforqed{\hbox{\rlap{$\sqcap$}$\sqcup$}}
\def\qed{\ifmmode\squareforqed\else{\unskip\nobreak\hfil
\penalty50\hskip1em\null\nobreak\hfil\squareforqed
\parfillskip=0pt\finalhyphendemerits=0\endgraf}\fi}

\newcommand\sgame{{\sf G}}

\newcommand\sinit{{\bar{s}}}

\newcommand\act{A}
\newcommand\dist{{\mathit{Dist}}}
\newcommand\Dist{{\dist}}

\def\Prb{{\mathit{Pr}}}

\def\future{{{\mathtt F}\ }}

\def\calP{{\mathtt P}}

\def\calR{{\mathtt R}}

\newcommand{\gametuple}{\langle S_\Diamond, S_\Box, \sinit, \act, \delta\rangle}

\def\sinit{{\overline{s}}}

\newcommand{\M}{{\mathcal{M}}}

\newcommand{\ipaths}{\mathit{IPath}}
\newcommand{\fpaths}{\mathit{FPath}}

\newcommand{\A}{\mathcal{A}}

\renewcommand{\leq}{\leqslant}
\renewcommand{\geq}{\geqslant}
\renewcommand{\le}{\leqslant}
\renewcommand{\ge}{\geqslant}

\newcommand{\sgamestrat}{\sgame^\mstrat}

%% file: fig_running_game.tex
\begin{tikzpicture}[scale=2]
  \setlength\tabcolsep{3pt}

\tikzstyle{every node}=[font=\small]
\tikzstyle{distr}=[inner sep=0mm, minimum size=1mm, draw, circle, fill]
\tikzstyle{statebox}=[inner sep=1mm, minimum size=5mm, draw]
\tikzstyle{statecircle}=[inner sep=0.3mm, minimum size=3mm, draw, circle]
\tikzstyle{statediamond}=[draw, diamond, inner sep=0.01mm, minimum size=6mm]
\tikzstyle{statediamondgoal}=[draw, diamond, inner sep=0.01mm, minimum size=6mm, fill=gray!25]
\tikzstyle{stategoal}=[inner sep=0.3mm, minimum size=3mm, draw, double, circle]
\tikzstyle{trarr}=[-latex, rounded corners]
\tikzstyle{lab}=[inner sep=0.7mm]


\node[lab] at (0.15,2.3) (init) {};

\node[statediamond] at (0.5,2.3) (s0) {$s_0$};
\node[statebox] at (2,2.3) (s1) {$s_1$};
\node[statediamond] at (3.5,2.3) (s2) {$s_2$};

\node[statediamond] at (0.5,1.4) (s3) {$s_3$};
\node[statebox] at (2,1.4) (s4) {$s_4$};
\node[statediamondgoal] at (3.5,1.4) (s5) {$s_5$};

\node[distr] at (2.1,1.9) (s1b) {};
\node[distr] at (2.1,1) (s4b) {};
\node[distr] at (3.5,1.1) (s5b) {};
\node[distr] at (0.8,1.7) (s3b) {};

\draw[trarr] (init) -- (s0);

\draw[trarr] (s0) -- (s1)
 node[above,pos=0.5]{$\mathit{east}$}
 node[distr,pos=0.5]{};

\draw[trarr] (s0) .. controls +(-0.4,-0.45) .. (s3)
 node[left,pos=0.5]{$\mathit{south}$}
 node[distr,pos=0.5]{};
\draw[-] (s3) -- (s3b)
 node[left,pos=1.1,xshift=-1pt,yshift=2pt]{$\mathit{north}$};
\draw[trarr] (s3b) -- (s0) 
 node[right,pos=0.5]{$0.7$}
 node[coordinate,pos=0.2](s3b1){};
\draw[trarr] (s3b) -- (s4) 
 node[above,pos=0.5]{$0.3$}
 node[coordinate,pos=0.125](s3b2){};
\path[-] (s3b1) edge [bend left] (s3b2);

\draw[-] (s1) -- (s1b)
 node[right,pos=0.4]{$\mathit{impede}$};
\draw[trarr] (s1b) .. controls +(-0.2,-0.1) .. (s0) 
 node[above,pos=0.5, yshift=2pt]{$0.75$}
 node[coordinate,pos=0.125](s1b1){};
\draw[trarr] (s1b) .. controls +(0.25,-0.05) .. (s2) 
 node[below,pos=0.80, xshift=5pt]{$0.25$}
 node[coordinate,pos=0.125](s1b2){};
\path[-] (s1b1) edge [bend right] (s1b2);

\draw[trarr] (s1) -- (s2)
 node[above,pos=0.5]{$\mathit{pass}$}
 node[distr,pos=0.5]{};

\draw[trarr] (s2) -- (s5)
 node[right,pos=0.5]{$\mathit{south}$}
 node[distr,pos=0.5]{};

\draw[trarr] (s3) -- +(0.8,0) -- (s4)
 node[below,pos=0]{$\mathit{east}$}
 node[distr,pos=0]{};

\draw[-] (s4) -- (s4b)
 node[right,pos=0.4]{$\mathit{impede}$};
\draw[trarr] (s4b) .. controls +(-0.2,-0.1) .. (s3) 
 node[above,pos=0.5]{$0.6$}
 node[coordinate,pos=0.125](s4b1){};
\draw[trarr] (s4b) .. controls +(0.25,-0.05) .. (s5)
 node[below,pos=0.6]{$0.4$}
 node[coordinate,pos=0.125](s4b2){};
\path[-] (s4b1) edge [bend right] (s4b2);

\draw[trarr] (s4) -- (s5)
 node[above,pos=0.5]{$\mathit{pass}$}
 node[distr,pos=0.5]{};

\draw[-] (s5) .. controls +(-0.3,-0.3) .. (s5b)
 node[below,pos=1]{$\mathit{done}$};
\draw[trarr] (s5b) .. controls +(0.3,0.0) .. (s5) ;

\end{tikzpicture}

%% file: fig_transformation.tex
\begin{figure}[!t]

\begin{center}
\begin{tikzpicture}[scale=1]
\tikzstyle{every node}=[font=\small]
\tikzstyle{distr}=[inner sep=0mm, minimum size=1mm, draw, circle, fill]
\tikzstyle{statebox}=[inner sep=1mm, minimum size=5mm, draw]
\tikzstyle{statecircle}=[inner sep=0.3mm, minimum size=3mm, draw, circle]
\tikzstyle{statediamond}=[draw, diamond, inner sep=0.01mm, minimum size=6mm]
\tikzstyle{trarr}=[-latex, rounded corners]
\tikzstyle{lab}=[inner sep=0.7mm]
\tikzstyle{bigarrow}=[draw, single arrow, minimum height=10ex,minimum width=1ex]
\node[statediamond] at (-1,-2) (s0) {$s$};

\node[distr] at (0,-1.65) (s0a1) {};
\node[rotate=90] at (0,-2) {...};
\node[distr] at (0,-2.35) (s0ak) {};
  

\draw[trarr] (s0) -- (s0a1) node[lab, pos=0.4,above] {$a_1$} -- +(0.8,0.25) node[coordinate,pos=0.4](s0a11){};
\draw[trarr] (s0a1) -- +(0.8,-0.25) node[coordinate,pos=0.4](s0a12){};
\draw[trarr] (s0) -- (s0ak) node[lab, pos=0.4,below] {$a_k$} -- +(0.8,0.25) node[coordinate,pos=0.4](s0ak1){};
\draw[trarr] (s0ak) -- +(0.8,-0.25) node[coordinate,pos=0.4](s0ak2){};
\path[-] (s0a11) edge [bend left] (s0a12);
\path[-] (s0ak1) edge [bend left] (s0ak2);


\node[statediamond] at (3,-2) (s0_app) {$s$};
\node[statediamond] at (5,-1.25) (g_bar_1) {$s_1$};
\node[statediamond] at (5,-2.75) (g_bar_n) {$s_m$};
\draw[very thick, dotted] ($(g_bar_1.south)+(0,-.1)$) -- ($(g_bar_n.north)+(0,.1)$);

\node[statediamond] at (7,-1.4) (g1) {$s'_1$};
\node[statediamond] at (7,-2.6) (gk) {$s'_2$};


\node[distr] at (4,-2) (s0_app_g) {};
\draw[trarr] (s0_app) -- (s0_app_g) -- (g_bar_1) node[lab, pos=0.5,above left] {$p_1$} node[coordinate,pos=0.3](g_bar_11){};
\draw[trarr] (s0_app_g) -- (g_bar_n) node[lab, pos=0.5,below left] {$p_m$} node[coordinate,pos=0.3](g_bar_n1){};
\path[-] (g_bar_11) edge [bend left] (g_bar_n1);

\draw[trarr] (g_bar_1) -- (g1) node[pos=0.4, distr] {} node[pos=0.2, above] {$b_1$};
\draw[trarr] (g_bar_1) -- (gk) node[pos=0.4, distr] {} node[lab, pos=0.1,below] {$b_2$};
  
\draw[trarr] (g_bar_n) -- (g1) node[pos=0.4, distr] {} node[pos=0.1, above] {$b_1$};
\draw[trarr] (g_bar_n) -- (gk) node[pos=0.4, distr] {} node[lab, pos=0.2,below] {$b_2$};

\node[distr] at (8,-1.1) (g1_a1) {};
\node[distr] at (8,-1.7) (g1_ak) {};
\node[rotate=90] at (8,-1.4) {...};
  
\draw[trarr] (g1) -- (g1_a1) node[lab, pos=0.4,above] {$a_1$} -- +(0.8,0.15) node[coordinate,pos=0.4](g1_a11){};
\draw[trarr] (g1_a1) -- +(0.8,-0.15) node[coordinate,pos=0.4](g1_a12){};
\draw[trarr] (g1) -- (g1_ak) node[lab, pos=0.4,below] {$a_k$} -- +(0.8,0.15) node[coordinate,pos=0.4](g1_ak1){};
\draw[trarr] (g1_ak) -- +(0.8,-0.15) node[coordinate,pos=0.4](g1_ak2){};
\path[-] (g1_a11) edge [bend left] (g1_a12);
\path[-] (g1_ak1) edge [bend left] (g1_ak2);

\node[distr] at (8,-2.3) (gk_a1) {};
\node[distr] at (8,-2.9) (gk_ak) {};
\node[rotate=90] at (8,-2.6) {...};

\draw[trarr] (gk) -- (gk_a1) node[lab, pos=0.4,above] {$a_1$} -- +(0.8,0.15) node[coordinate,pos=0.4](gk_a11){};
\draw[trarr] (gk_a1) -- +(0.8,-0.15) node[coordinate,pos=0.4](gk_a12){};
\draw[trarr] (gk) -- (gk_ak) node[lab, pos=0.4,below] {$a_k$} -- +(0.8,0.15) node[coordinate,pos=0.4](gk_ak1){};
\draw[trarr] (gk_ak) -- +(0.8,-0.15) node[coordinate,pos=0.4](gk_ak2){};
\path[-] (gk_a11) edge [bend left] (gk_a12);
\path[-] (gk_ak1) edge [bend left] (gk_ak2);

\end{tikzpicture}
\end{center}
\caption{A node in the original game $\sgame$ (left), and the corresponding nodes in the transformed game $\sgame'$ for approximating randomised multi-strategies (right, see \sectref{sec:approximation}).}
\label{fig:approximation}
\end{figure}

%% file: fig_counterexample.tex
\begin{figure}[!h]
\centering

\begin{tikzpicture}[scale=2]
\tikzstyle{every node}=[font=\small]
\tikzstyle{distr}=[inner sep=0mm, minimum size=1mm, draw, circle, fill]
\tikzstyle{statebox}=[inner sep=1mm, minimum size=5mm, draw]
\tikzstyle{stateboxgoal}=[inner sep=1mm, minimum size=5mm, draw, fill=gray!25]
\tikzstyle{statecircle}=[inner sep=0.3mm, minimum size=3mm, draw, circle]
\tikzstyle{statediamond}=[draw, diamond, inner sep=0.01mm, minimum size=6mm]
\tikzstyle{trarr}=[-latex, rounded corners]
\tikzstyle{lab}=[inner sep=0.7mm]
\node[statediamond] at (0,0) (s0) {$s_0$};
\node[statediamond] at (1,0) (s1) {$s_1$};
\node[statebox] at (0,-0.5) (s2) {$s_2$};
\node[stateboxgoal] at (2,0) (s3) {$s_3$};
\node[statebox] at (1,-0.5) (s4) {$s_4$};

\draw[trarr] (s0) -- (s1) node[distr,pos=0.5]{} node[pos=0.2, above] {$b$};
\draw[trarr] (s0) -- (s2) node[distr,pos=0.5]{} node[pos=0.2, left] {$c$};
\draw[trarr] (s1) -- (s3) node[distr,pos=0.5]{} node[pos=0.2, above] {$d$};
\draw[trarr] (s1) -- (s4) node[distr,pos=0.5]{} node[pos=0.2, left] {$e$};

\draw[trarr]  (s2) .. controls +(0.5,-0.25) and +(0.5,0.25) .. (s2)
 node[distr,pos=0.5]{};

\draw[trarr]  (s3) .. controls +(0.5,-0.25) and +(0.5,0.25) .. (s3)
 node[distr,pos=0.5]{};

\draw[trarr]  (s4) .. controls +(0.5,-0.25) and +(0.5,0.25) .. (s4)
 node[distr,pos=0.5]{};


\node[draw] at (3.25,-0.25) {\scriptsize\begin{tabular}{cc}
action & penalty\\
$b$ & $0$\\
$c$ & $1$\\
$d$ & $0$\\
$e$ & $1$
\end{tabular}};

\end{tikzpicture}

\caption{Counterexample for \thmref{2_set_rand_subset} in case of dynamic penalties.}
\label{fig:chain_condition_cex}

\end{figure}

%% file: fig_hard_rand1.tex
\begin{tikzpicture}[scale=1.2]
\tikzstyle{every node}=[font=\small]
\tikzstyle{distr}=[inner sep=0mm, minimum size=1mm, draw, circle, fill]
\tikzstyle{statebox}=[inner sep=1mm, minimum size=5mm, draw]
\tikzstyle{statecircle}=[inner sep=0.3mm, minimum size=3mm, draw, circle]
\tikzstyle{statediamond}=[draw, diamond, inner sep=0.01mm, minimum size=6mm]
\tikzstyle{trarr}=[-latex, rounded corners]
\tikzstyle{lab}=[inner sep=0.7mm]

\draw[very thick, dotted] (1.5,0.3) -- (1.5,-0.3);
\node[statediamond] at (0,0) (s0) {$s_0$};
\node[distr] at (1,0) (s0d) {};

\node[statediamond] at (2,1) (t1) {$t_1$};
\node[statediamond] at (4,1) (tp1) {$t'_1$};

\node[distr] at (5,1) (tp1d) {};

\node[statediamond] at (6,1.4) (fin1a) {$t''_1$};
\node[statediamond] at (7,1.4) (fin1b) {$\top$};
\node[statediamond] at (6,0.6) (bot1) {$\bot$};

\node[statediamond] at (2,-1) (tn) {$t_n$};
\node[statediamond] at (4,-1) (tpn) {$t'_n$};

\node[distr] at (5,-1) (tpnd) {};

\node[statediamond] at (6,-0.6) (finna) {$t''_1$};
\node[statediamond] at (7,-0.6) (finnb) {$\top$};
\node[statediamond] at (6,-1.4) (botn) {$\bot$};

\draw[trarr] (s0) -- (s0d) -- (t1) node[above,pos=0.5] {$1/n$};
\draw[trarr] (s0d) -- (tn) node[below,pos=0.5] {$1/n$};

\draw[trarr] (t1) -- (tp1) node[below,pos=0.5] {$a_1$};
\draw[trarr] (t1) |- (bot1.south west) node[below,pos=0.75] {$b_1, \underline{w_1}$};
\draw[trarr] (tp1) -- (tp1d) node[above,pos=0.5] {$d_1$} -- (fin1a) node[above,pos=0.75] {$v_1$} -- (fin1b) node[above,pos=0.5] {$\overline{1}$};
\draw[trarr] (tp1d) -- (bot1) node[below,pos=0.25] {$1-v_1$};
\draw[trarr] (tp1) -- +(0,0.4) -| (t1) node[above,pos=0.25] {$c_1, \underline{2^{3m}\cdot w_1}$};

\draw[trarr] (tn) -- (tpn) node[below,pos=0.5] {$a_n$};
\draw[trarr] (tn) |- (botn.south west) node[below,pos=0.75] {$b_n, \underline{w_n}$};
\draw[trarr] (tpn) -- (tpnd) node[above,pos=0.5] {$d_n$} -- (finna) node[above,pos=0.75] {$v_n$} -- (finnb) node[above,pos=0.5] {$\overline{1}$};
\draw[trarr] (tpnd) -- (botn) node[below,pos=0.25] {$1-v_n$};
\draw[trarr] (tpn) -- +(0,0.4) -| (tn) node[above,pos=0.25] {$c_n, \underline{2^{3m}\cdot w_n}$};
\end{tikzpicture}

%% file: fig_hard_rand2.tex
\begin{tikzpicture}[scale=1.2]
\tikzstyle{every node}=[font=\small]
\tikzstyle{distr}=[inner sep=0mm, minimum size=1mm, draw, circle, fill]
\tikzstyle{statebox}=[inner sep=1mm, minimum size=5mm, draw]
\tikzstyle{statecircle}=[inner sep=0.3mm, minimum size=3mm, draw, circle]
\tikzstyle{statediamond}=[draw, diamond, inner sep=0.01mm, minimum size=6mm]
\tikzstyle{trarr}=[-latex, rounded corners]
\tikzstyle{lab}=[inner sep=0.7mm]

\draw[very thick, dotted] (1.35,0.2) -- (1.35,-0.2);
\node[statediamond] at (0,0) (s0) {$s_0$};
\node[distr] at (1,0) (s0d) {};

\node[statediamond] at (2,1) (t1) {$t_1$};
\node[distr] at (3,1) (t1d) {};

\node[statediamond] at (4,1.4) (fin1a) {};
\node[statediamond] at (5,1.4) (fin1b) {$\top$};
\node[statediamond] at (4,0.6) (bot1) {$\bot$};

\node[statediamond] at (2,-1) (tn) {$t_n$};

\node[distr] at (3,-1) (tnd) {};

\node[statediamond] at (4,-0.6) (finna) {};
\node[statediamond] at (5,-0.6) (finnb) {$\top$};
\node[statediamond] at (4,-1.4) (botn) {$\bot$};

\draw[trarr] (s0) -- (s0d) -- (t1) node[left,pos=0.6] {$1/n$};
\draw[trarr] (s0d) -- (tn) node[left,pos=0.6] {$1/n$};

\draw[trarr] (t1) -- (t1d) node[below,pos=0.5] {$a_1$} -- (fin1a) node[above,pos=0.5] {$v_1$} -- (fin1b) node[above,pos=0.5] {$\overline{1}$};
\draw[trarr] (t1d) -- (bot1) node[below,pos=0.5] {$1-v_1$};
\draw[trarr] (t1) -- ++(-0.3,0.5) -- ++(0.6,0) node[distr,pos=0.5] {} node[above,pos=0.25] {$b_1, \underline{w_1}$} -- (t1);

\draw[trarr] (tn) -- (tnd) node[below,pos=0.5] {$a_n$} -- (finna) node[above,pos=0.5] {$v_n$} -- (finnb) node[above,pos=0.5] {$\overline{1}$};
\draw[trarr] (tnd) -- (botn) node[below,pos=0.5] {$1-v_n$};
\draw[trarr] (tn) -- ++(-0.3,0.5) -- ++(0.6,0) node[distr,pos=0.5] {} node[above,pos=0.25] {$b_n, \underline{w_n}$} -- (tn);

\end{tikzpicture}

%% file: fig_sqrt.tex
\begin{tikzpicture}[scale=1.2]
\tikzstyle{every node}=[font=\small]
\tikzstyle{distr}=[inner sep=0mm, minimum size=1mm, draw, circle, fill]
\tikzstyle{statebox}=[inner sep=1mm, minimum size=5mm, draw]
\tikzstyle{statecircle}=[inner sep=0.3mm, minimum size=3mm, draw, circle]
\tikzstyle{statediamond}=[draw, diamond, inner sep=0.01mm, minimum size=6mm]
\tikzstyle{trarr}=[semithick, -latex, rounded corners]
\tikzstyle{lab}=[inner sep=0.7mm]

\draw[very thick, dotted] (1.9,-0.2) -- (1.9,-0.8);
\draw[very thick, dotted] (6,-0.4) -- (6,-1);
\node[statebox] at (0.6,-0.5) (s0) {$s_0$};
\node[distr] at (1.3,-0.5) (s0d) {};

\node[statebox] at (2,1) (ta1) {$t_1$};
\node[distr] at (3,1.5) (ta1d1) {};
\node[distr] at (3,0.2) (ta1d2) {};
\node[statediamond] at (4,0.2) (tb1) {$t'_1$};
\node[distr] at (5,0.4) (tb1d1) {};
\node[distr] at (5,0) (tb1d2) {};
\node[statediamond] at (6,0.6) (tc1) {$\bar t_1$};
\node[distr] at (7,0.8) (tc1d1) {};
\node[distr] at (7,0.6) (tc1d2) {};
\node[statebox] at (8,1) (td1) {};
\node[distr] at (9.5,1) (td1d) {};
\node[statebox] at (10,1) (te1) {$\bot$};

\node[statebox] at (2,-2) (tan) {$t_n$};
\node[distr] at (3,-1.5) (tand1) {};
\node[distr] at (3,-2.8) (tand2) {};
\node[statediamond] at (4,-2.8) (tbn) {$t'_n$};
\node[distr] at (5,-2.6) (tbnd1) {};
\node[distr] at (5,-3) (tbnd2) {};
\node[statediamond] at (6,-2.4) (tcn) {$\bar t_n$};
\node[distr] at (7,-2.2) (tcnd1) {};
\node[distr] at (7,-2.4) (tcnd2) {};
\node[statebox] at (8,-2) (tdn) {};
\node[distr] at (9.5,-2) (tdnd) {};
\node[statebox] at (10,-2) (ten) {$\bot$};

\draw[->] (s0) -- (s0d) -- (ta1) node[left,pos=0.5] {$1/n$};
\draw[->] (ta1) -- (ta1d1) -| (te1) node[above,pos=0.15] {$1-x_1$};
\draw[->] (ta1d1) -- (td1) node[below,pos=0.3] {$x_1$};
\draw[->] (ta1) -- (ta1d2) -- (tb1);
\draw[->] (tb1) -- (tb1d1) node[above,pos=0.5] {$c'_1$}  -- (tc1);
\draw[->] (tb1.south east) -- (tb1d2) node[below,pos=0.5] {$a'_1,\underline{1}$} -| (te1.south east);
\draw[->] (tc1) -- (tc1d1) node[above,pos=0.5] {$\bar c_1$} -- (td1);
\draw[->] (tc1) -- (tc1d2) node[below,pos=0.5] {$\bar a_1,\underline{1}$} -| (te1.south);
\draw[->] (td1) -- (td1d) node[above,pos=0.5] {$b_1,1/x_1$} -- (te1);

\draw[->] (s0d) -- (tan) node[left,pos=0.5] {$1/n$};
\draw[->] (tan) -- (tand1) -| (ten) node[above,pos=0.15] {$1-x_n$};
\draw[->] (tand1) -- (tdn) node[below,pos=0.3] {$x_n$};
\draw[->] (tan) -- (tand2) -- (tbn);
\draw[->] (tbn) -- (tbnd1) node[above,pos=0.5] {$c'_1$}  -- (tcn);
\draw[->] (tbn.south east) -- (tbnd2) node[below,pos=0.5] {$a'_n,\underline{1}$} -| (ten.south east);
\draw[->] (tcn) -- (tcnd1) node[above,pos=0.5] {$\bar c_n$}  -- (tdn);
\draw[->] (tcn) -- (tcnd2) node[below,pos=0.5] {$\bar a_n,\underline{1}$} -| (ten.south);
\draw[->] (tdn) -- (tdnd) node[above,pos=0.5] {$b_n,1/x_n$} -- (ten);

\end{tikzpicture}